\documentclass[unsortedaddress,superscriptaddress,pra,notitlepage]{revtex4}

\usepackage{amsfonts}
\usepackage[latin2]{inputenc}
\usepackage{t1enc}
\usepackage{amssymb}
\usepackage{amsmath}
\usepackage[mathscr]{eucal}
\usepackage{amsthm}
\usepackage{array}
\usepackage{color}
\usepackage[hidelinks]{hyperref}

\usepackage{graphicx}
\usepackage{xcolor}
\usepackage{tikz}

\usetikzlibrary{calc,decorations.pathreplacing}
\usetikzlibrary{arrows,shapes}

\theoremstyle{definition} 
\newtheorem{defin}{Definition}[section]

\newtheorem{thm}[defin]{Theorem}

\newtheorem{rem}[defin]{Remark}
\newtheorem{ex}[defin]{Example}
\newtheorem{cor}[defin]{Corollary}

\newtheorem{lemma}[defin]{Lemma}
\newtheorem{prop}[defin]{Proposition}

\renewcommand{\theHWexercise}{\theHWexercise$^*$}

\def\vfi{\varphi}
\def\theta{\vartheta}
\def\hil{{\mathcal H}}

\def\A{{\mathcal A}}

\def\B{{\mathcal B}}

\def\E{{\mathcal E}}

\def\M{\mathcal{M}}

\def\P{{\mathcal P}}
\def\S{{\mathcal S}}
\def\T{{\mathcal T}}

\def\X{{\mathcal X}}

\def\half{\frac{1}{2}}
\def\iff{\Longleftrightarrow}
\def\imp{\Longrightarrow}

\def\bN{\mathbb{N}}
\def\bC{\mathbb{C}}
\def\bP{\mathbb{P}}

\def\bR{\mathbb{R}}

\def\bT{\mathbb{T}}
\def\bz{\left(}
\def\jz{\right)}

\def\inv{^{-1}}

\def\egy{\mathbf 1}

\def\rho{\varrho}
\def\povm{\mathrm{POVM}}

\def\pvm{\mathrm{PVM}}

\def\nn{\nonumber}

\def\nw{^{*}}

\def\meas{\mathrm{meas}}
\def\test{\mathrm{test}}

\def\scli{\underline{\mathrm{sc}}}

\def\oll{\overline}

\def\p{_{\ge 0}}

\def\valt{\cdot}

\newcommand{\ki}[1]{\textit{\textit{#1}}}

\newcommand{\s}{\mbox{ }}
\newcommand{\ds}{\mbox{ }\mbox{ }}

\newcommand{\inner}[2]{\left\langle #1 , #2\right\rangle}

\newcommand{\vecc}[1]{\underline{#1}}

\newcommand{\ket}[1]{\left|#1\right\rangle}
\newcommand{\diad}[2]{\left|#1\right\rangle\!\left\langle #2\right|}

\newcommand{\pr}[1]{\diad{#1}{#1}}

\newcommand{\vertleq}{\rotatebox{90}{$\,\ge$}}

\makeatletter
\renewcommand{\p@enumii}{}
\makeatother

\DeclareMathOperator{\Tr}{Tr}

\DeclareMathOperator{\supp}{supp}

\DeclareMathOperator{\ran}{ran}

\DeclareMathOperator{\divv}{\Delta}

\begin{document}

\title{Test-measured R\'enyi divergences}

\author{Mil\'an Mosonyi}
\email{milan.mosonyi@gmail.com}

\affiliation{
MTA-BME Lend\"ulet Quantum Information Theory Research Group
}

\affiliation{
Department of Analysis, Institute of Mathematics,\\
 Budapest University of Technology and Economics, M\H uegyetem rkp.~3., H-1111 Budapest, Hungary
}

\author{Fumio Hiai}
\email{hiai.fumio@gmail.com}

\affiliation{
Graduate School of Information Sciences, Tohoku University, \\
Aoba-ku, Sendai 980-8579, Japan
}

\begin{abstract}
\centerline{\textbf{Abstract}}
\vspace{.3cm}

One possibility of defining a quantum R\'enyi $\alpha$-divergence of two quantum 
states is to optimize the classical R\'enyi $\alpha$-divergence 
of their post-measurement probability distributions over all possible measurements
(measured R\'enyi divergence), and maybe regularize these quantities over 
multiple copies of the two states (regularized measured R\'enyi $\alpha$-divergence).
A key observation behind the theorem for the strong converse exponent of 
asymptotic binary quantum state discrimination 
is that the regularized measured R\'enyi $\alpha$-divergence coincides with the sandwiched R\'enyi 
$\alpha$-divergence when $\alpha>1$. Moreover, it also follows from the same theorem that 
to achieve this, it is sufficient to consider $2$-outcome measurements 
(tests) for any 
number of copies (this is somewhat surprising, as achieving the measured R\'enyi 
$\alpha$-divergence
for $n$ copies might require a number of measurement outcomes that diverges in $n$, in general). 
In view of this, it seems natural to expect the same when $\alpha<1$; however, 
we show that this is not the case.
In fact, we show that even for commuting states (classical case) the regularized 
quantity attainable using $2$-outcome 
measurements is in general strictly smaller than the R\'enyi $\alpha$-divergence (which is unique in the 
classical case). In the general quantum case this shows that the above ``regularized test-measured''
R\'enyi $\alpha$-divergence is not even a quantum extension of the classical R\'enyi divergence when $\alpha<1$, in sharp contrast to the $\alpha>1$ case.
\end{abstract}

\maketitle

\section{Introduction}

It has been known for a long time in classical information theory that R\'enyi 
divergences and derived information quantities play a central role 
in quantifying the trade-off between the relevant operational quantities in many 
information theoretic problems, like source coding, channel coding, or state
discrimination (see, e.g., \cite{Csiszar}). Due to the non-commutativity of 
general quantum states, R\'enyi divergences can be extended to pairs of quantum states in 
infinitely many different ways; see, e.g., 
\cite{AD,BST,FawziFawzi2021,Hiai_fdiv_Springer,Jencova_NCLp,Jencova_NCLpII,Matsumoto_newfdiv,P86,Petz_QE_vN,PetzRuskai1998,Renyi_new,WWY} for various different extensions. Some of these extensions have similar operational roles as their classical counterpart
\cite{Aud,ANSzV,Hayashicq,HT14,HMO2,HiaiMosonyi2021,JOPS,MO,MO-cqconv,MO-cqconv-cc,Mosonyi_sc_2021,Nagaoka},
while others are interesting for their mathematical properties or as 
useful approximations to the operationally relevant quantities.

One natural way of defining a quantum R\'enyi $\alpha$-divergence of two quantum 
states $\rho$ and $\sigma$ is to optimize the classical R\'enyi $\alpha$-divergences 
of their post-measurement probability distributions over all possible measurements.
This leads to the notion of the 
\ki{measured R\'enyi $\alpha$-divergence} $D_{\alpha}^{\meas}(\rho\|\sigma)$,
which has nice mathematical properties, but no known direct operational interpretation
or closed-form expression. A variant of it with better properties is obtained by regularizing it 
over many copies of the states as 
$\oll{D}_{\alpha}^{\meas}(\rho\|\sigma):=\lim_{n\to+\infty}\frac{1}{n}D_{\alpha}^{\meas}(\rho^{\otimes n}\|\sigma^{\otimes n})$; this is called the 
\ki{regularized measured R\'enyi $\alpha$-divergence}.
Quite surprisingly, this admits a closed-form expression, as it turns out to be equal to the 
\ki{sandwiched R\'enyi $\alpha$-divergence} \cite{Renyi_new,WWY} for 
$\alpha\in[1/2,+\infty)$,
and a closed-form expression is also available for $\alpha\in(0,1/2)$. 
This was proved in the finite-dimensional case by asymptotic pinching
\cite{HT14,HP,MO}, and extended very recently 
to the infinite-dimensional case \cite{Mosonyi_sc_2021}, and more generally, to states
of nuclear $C^*$-algebras \cite{HiaiMosonyi2021}.
Moreover, this is one of the key observations behind the proof for the
expression of the strong converse exponent of
asymptotic binary quantum state discrimination in terms of the 
sandwiched R\'enyi $\alpha$-divergences with $\alpha>1$, given in 
\cite{MO,Mosonyi_sc_2021,HiaiMosonyi2021}.
It 
also follows from the strong converse theorem
that for $\alpha>1$, the regularized measured R\'enyi $\alpha$-divergence can be attained by 
considering only $2$-outcome measurements (tests) for each number of copies of the states. 
This is rather surprising when one takes into account that 
attaining $D_{\alpha}^{\meas}(\rho^{\otimes n}\|\sigma^{\otimes n})$ 
for $n$ copies of the states
requires in general 
a number of measurement outcomes that diverges in $n$
(this is true even for classical states, as can be seen easily by considering type decompositions).

This motivates the introduction of the 
\ki{(regularized) test-measured R\'enyi $\alpha$-divergences}
as the variants of the usual (regularized) measured R\'enyi $\alpha$-divergences
with only $2$-outcome measurements in their definitions. 
According to the above, 
the regularized measured and the regularized test-measured R\'enyi $\alpha$-divergences coincide
for $\alpha>1$, and hence it might seem reasonable to expect the same for $\alpha\in(0,1)$. 
Our main result in this paper is that this is not the case.
In fact, we show that even for commuting states (classical case) the regularized 
test-measured R\'enyi $\alpha$-divergence is in general 
strictly smaller than the R\'enyi $\alpha$-divergence (which is unique in the 
classical case) for every $\alpha\in(0,1)$. 
In the general quantum case this shows that the regularized test-measured
R\'enyi $\alpha$-divergence is not even a quantum extension of the classical R\'enyi divergence when $\alpha<1$, in contrast to the $\alpha>1$ case.

The structure of the paper is as follows. In Section \ref{sec:prelim}
we summarize the necessary preliminaries on R\'enyi divergences and the 
Hoeffding bound theorem of quantum state discrimination \cite{Hayashicq,Nagaoka}.
In Section \ref{sec:def} we introduce the test-measured R\'enyi $\alpha$-divergences, and 
two variants of their regularizations, and discuss some basic relations between 
these quantities and some previously studied R\'enyi divergences.
In Section \ref{sec:an1} we give an expression for one of the versions of the
regularized test-measured R\'enyi $\alpha$-divergence for $\alpha\in(0,1)$
in terms of the Hoeffding divergences, and 
use this to show that this version is strictly smaller than the standard (or Petz-type)
R\'enyi $\alpha$-divergence under very mild conditions on the states.
We use this result in Section \ref{sec:an2} to show that the other version 
of the regularized test-measured R\'enyi $\alpha$-divergence
is strictly smaller than the (unique)
R\'enyi $\alpha$-divergence for unequal commuting states with equal supports.
In particular, the results of Sections \ref{sec:an1} and \ref{sec:an2} together
yield that in the classical case (more precisely, for pairs of probability distributions on at least three points)
both versions of the 
regularized test-measured R\'enyi $\alpha$-divergence
are strictly smaller than the classical
R\'enyi $\alpha$-divergence for generic pairs of states and any $\alpha\in(0,1)$.
Moreover, we show that, somewhat surprisingly, the two different regularizations may give different values in the classical case.

In Appendix \ref{sec:SD} we 
explain a connection of our results in Section \ref{sec:an1}
with a very recent result by Salzmann and Datta on a variant of the quantum Hoeffding bound theorem 
\cite{Salzmann_Datta21}.
Finally, in Appendix \ref{sec:vN} we extend the main results in 
Section \ref{sec:an1} to the von Neumann algebra setting.

\section{Preliminaries}
\label{sec:prelim}

In the main body of the paper, $\hil$ will always denote a finite-dimensional Hilbert space.
We will use the notations $\B(\hil)$ for the set of linear operators on $\hil$, 
and $\B(\hil)\p$ for the set of positive semi-definite (PSD) operators.
Furthermore, $\S(\hil):=\{\rho\in\B(\hil)\p:\,\Tr\rho=1\}$ will denote the set of 
density operators, or states, 
$\bT(\hil):=\{T\in\B(\hil):\,0\le T\le I\}$ the set of \ki{tests} on $\hil$, 
and $\bP(\hil)$ the set of projections on $\hil$.

In our study of functions of pairs of density operators
we will often use that two commuting density operators can be diagonalized in the same 
orthonormal basis, and hence can be written as 
\begin{align}\label{commuting pair}
\rho=\sum_{\omega\in\Omega}\rho(\omega)\pr{\omega},\ds\ds\ds
\sigma=\sum_{\omega\in\Omega}\sigma(\omega)\pr{\omega},
\end{align}
with some orthonormal basis $(\ket{\omega})_{\omega\in\Omega}$, and probability density 
functions $\rho,\sigma$ on $\Omega$. We will refer to this setting as the classical case.

For a finite-dimensional Hilbert space $\hil$ and 
a finite set $\X$, let 
\begin{align*}
\povm(\hil,\X):=\left\{(M_x)_{x\in\X}:\,M_x\in\B(\hil)\p,\,x\in\X,\s\sum_x M_x=I\right\}
\end{align*}
denote the set of positive operator-valued measures (POVMs) on $\hil$ with outcome set $\X$, and let 
\begin{align*}
\pvm(\hil,\X):=\left\{(M_x)_{x\in\X}\in\povm(\hil,\X):\,M_x\in\bP(\hil),\,x\in\X\right\}
\end{align*}
denote the subset of projection-valued measures (PVMs). Let
\begin{align*}
\pvm_1(\hil):=\left\{M\in\pvm(\hil,[\dim\hil]):\,\Tr M_k=1,\,k=1,\ldots,\dim\hil\right\}
\end{align*}
be the set of rank $1$ PVMs on $\hil$, i.e, the set of measurements in an orthonormal basis 
of $\hil$. 
For any $M\in\povm(\hil,\X)$, let 
\begin{align*}
\M(A):=\sum_{x\in\X}(\Tr M_x A)\egy_{\{x\}}\in\bC^{\X},\ds\ds\ds A\in\B(\hil),
\end{align*}
where $\egy_{\{x\}}$ is the indicator function of the singleton $\{x\}$. If 
$\rho\in\S(\hil)$ is a state then $\M(\rho)$ is the post-measurement 
probability distribution.
\medskip

For two probability density functions $\rho,\sigma$ on some finite set $\Omega$, 
and $\alpha\in(0,+\infty)\setminus\{1\}$, let 
\begin{align*}
D_{\alpha}(\rho\|\sigma):=
\begin{cases}
\frac{1}{\alpha-1}\log\sum_{\omega\in\Omega}\rho(\omega)^{\alpha}\sigma(\omega)^{1-\alpha},&\alpha\in(0,1)\text{ or }\supp\rho\subseteq\supp\sigma,\\
+\infty,&\text{ otherwise},
\end{cases}
\end{align*}
denote the (classical) R\'enyi $\alpha$-divergence of $\rho$ and $\sigma$. 
For $\alpha=1$ we have 
\begin{align*}
D_1(\rho\|\sigma)&:=
\lim_{\alpha\to 1}D_{\alpha}(\rho\|\sigma)=
D(\rho\|\sigma):=
\begin{cases}
\sum_{\omega\in\Omega}\rho(\omega)(\log\rho(\omega)-\log\sigma(\omega)),
&\supp\rho\subseteq\supp\sigma,\\
+\infty,&\text{ otherwise},
\end{cases}
\end{align*}
where $D(\rho\|\sigma)$ is the \ki{Kullback-Leibler divergence}, or 
\ki{relative entropy} of $\rho$ and $\sigma$.

There are various extensions of the R\'enyi $\alpha$-divergences to pairs of quantum states. 
Motivated e.g., by the notion of a quantum $f$-divergence introduced in 
\cite{Matsumoto_newfdiv}, we consider the following:

\begin{defin}\label{def:qRenyi}
For $\alpha\in(0,1)\cup(1,+\infty)$, a function
\begin{align*}
D_{\alpha}^q:\,\cup_{d\in\bN}\bz\S(\bC^d)\times\S(\bC^d)\jz\to\bR
\end{align*}
is a \ki{quantum R\'enyi $\alpha$-divergence} if it is invariant under isometries, i.e., 
for any $\rho,\sigma\in\S(\bC^d)$ and any isometry $V:\,\bC^d\to\bC^{d'}$, 
\begin{align*}
D_{\alpha}^q\bz V\rho V^*\|V\sigma V^*\jz=D_{\alpha}^q(\rho\|\sigma),
\end{align*}
and it reduces to the classical R\'enyi $\alpha$-divergence on commuting states, i.e., 
if $\rho=\sum_{i=1}^dp_i\pr{i}$ and $\sigma=\sum_{i=1}^d q_i\pr{i}$ are diagonal in the same orthonormal basis then 
\begin{align*}
D_{\alpha}^q(\rho\|\sigma)=D_{\alpha}\bz(p_i)_{i=1}^d\|(q_i)_{i=1}^d\jz.
\end{align*}
\end{defin}

It is clear that any quantum R\'enyi $\alpha$-divergence can be uniquely extended to pairs of 
density operators $\rho,\sigma$ on an arbitrary finite-dimensional Hilbert space $\hil$ by 
mapping $\hil$ into some $\bC^d$ with an isometry $V$, and defining 
$D_{\alpha}^q(\rho\|\sigma):=D_{\alpha}^q(V\rho V^*\|V\sigma V^*)$. 

Two particularly important families of quantum R\'enyi divergences are the
\ki{standard (or Petz-type) R\'enyi $\alpha$-divergences} \cite{P86}, given for 
$\rho,\sigma\in\S(\hil)$, and $\alpha\in(0,+\infty)\setminus\{1\}$, as
\begin{align}\label{standard Renyi def}
D_{\alpha}(\rho\|\sigma):=
\begin{cases}
\frac{1}{\alpha-1}\log\Tr\rho^{\alpha}\sigma^{1-\alpha},&\alpha\in(0,1)\text{ or }\supp\rho\subseteq\supp\sigma,\\
+\infty,&\text{ otherwise},
\end{cases}
\end{align}
and the \ki{sandwiched R\'enyi $\alpha$-divergences} \cite{Renyi_new,WWY}
\begin{align*}
D_{\alpha}\nw(\rho\|\sigma):=
\begin{cases}
\frac{1}{\alpha-1}\log\Tr\bz\rho^{1/2}\sigma^{\frac{1-\alpha}{\alpha}}\rho^{1/2}\jz^{\alpha},&\alpha\in(0,1)\text{ or }\supp\rho\subseteq\supp\sigma,\\
+\infty,&\text{ otherwise}.
\end{cases}
\end{align*}
Here and henceforth we follow the convention that for a PSD operator $A\in\B(\hil)\p$
with spectral decomposition $A=\sum_aaP_a$, where $P_a$ is the projection onto
$\{\psi\in\hil:\,A\psi=a\psi\}$, $a\in[0,+\infty)$, real powers of $A$ are defined as
$A^x:=\sum_{a>0}a^xP_a$. In particular, $A^0$ is the projection onto the support of $A$. With this
convention, 
\begin{align*}
D_0(\rho\|\sigma):=\lim_{\alpha\searrow 0}D_{\alpha}(\rho\|\sigma)=-\log\Tr\rho^0\sigma.
\end{align*}
\medskip

In what follows, we consider further functions of pairs of density operators $\rho$ and $\sigma$. To avoid trivial pathological cases, we will always implicitly assume that 
\begin{align*}
\rho\not\perp\sigma,\ds\ds\ds\text{i.e.,}\ds\ds\ds \Tr\rho\sigma>0.
\end{align*}

For $\rho,\sigma\in\S(\hil)$, let us introduce
\begin{align*}
\psi(\alpha)&:=\psi(\rho\|\sigma|\alpha):=\log\Tr\rho^{\alpha}\sigma^{1-\alpha},\ds\ds\ds\ds\ds\ds\alpha\in\bR,
\end{align*}
so that $D_{\alpha}(\rho\|\sigma)=\frac{\psi(\alpha)}{\alpha-1}$ 
if $\alpha\in(0,1)$ or
$\rho^0\le\sigma^0$ (otherwise $D_\alpha(\rho\|\sigma)=+\infty$).

We will need the following:

\begin{lemma}\label{lemma:Dalpha monotone}
Let $\rho,\sigma\in\S(\hil)$.
The function $\psi(\rho\|\sigma|\valt)$ is 
convex and real analytic on $\bR$, 
and it is non-positive on $[0,1]$.
The functions
$\alpha\mapsto D_{\alpha}(\rho\|\sigma)$ and
$\alpha\mapsto D_{\alpha}\nw(\rho\|\sigma)$ are
non-negative and increasing on $(0,1)\cup(1,+\infty)$
with 
\begin{align*}
D_1(\rho\|\sigma)&:=
D_1\nw(\rho\|\sigma):=
\lim_{\alpha\to 1}D_{\alpha}(\rho\|\sigma)
=
\lim_{\alpha\to 1}D_{\alpha}\nw(\rho\|\sigma)\\
&=D(\rho\|\sigma):=
\begin{cases}
\Tr\rho(\log\rho-\log\sigma),&\rho^0\le\sigma^0,\\
+\infty,&\text{otherwise},
\end{cases}
\end{align*}
being the \ki{relative entropy} of $\rho$ and $\sigma$ \cite{Umegaki}.
Moreover, 
$\alpha\mapsto D_{\alpha}(\rho\|\sigma)$ is strictly increasing on $(0,1)$
unless 
\begin{align}\label{constant Renyi char}
\rho=\sum_{i=1}^m\kappa s_i P_i,
\end{align}
for some $\kappa>0$ and projections $P_i\le Q_i$, $i\in[m]$,
where $\sigma=\sum_{i=1}^ms_iQ_i$ is the spectral decomposition of $\sigma$.
In this latter case $D_{\alpha}(\rho\|\sigma)=\log\kappa$, $\alpha\in(0,+\infty)$.
\end{lemma}
\begin{proof}
Real analyticity of $\psi$ on $\bR$ is easy to see, 
and its convexity can be easily verified by simply computing its second derivative;
see, e.g., \cite[Lemma 3.2]{HMO2}.
Non-positivity of $\psi$ on $[0,1]$ follows from the above and that $\psi(0)\le 0$, $\psi(1)\le 0$.
Convexity of $\psi$ implies that
\begin{align*}
D_{\alpha}(\rho\|\sigma)=\frac{\psi(\alpha)-\psi(1)}{\alpha-1}+\frac{\psi(1)}{\alpha-1}
\end{align*}
is strictly increasing on $(0,1)$, unless $\psi(1)=0$, i.e., $\rho^0\le\sigma^0$, and 
$\psi$ is affine. The characterization 
of $\alpha\mapsto D_{\alpha}(\rho\|\sigma)$ not being strictly increasing 
then follows from the characterization of 
$\psi$ being affine given in \cite[Lemma 3.2]{HMO2}. For the assertions about the sandwiched R\'enyi
divergence, see \cite{Renyi_new}.
\end{proof}
\begin{rem}
The above proof also yields that $\alpha\mapsto\frac{\psi(\alpha)}{\alpha-1}$ 
($=D_{\alpha}(\rho\|\sigma)$ for
$\alpha\in(0,1)$) is strictly increasing on $(0,+\infty)$ unless  
\eqref{constant Renyi char} holds.
\end{rem}

We will also need the Legendre transforms
\begin{align}
\vfi(c)&:=\max_{\alpha\in[0,1]}\{c(\alpha-1)-\psi(\alpha)\},\ds\ds\ds c\in\bR,
\label{phi def}\\
\vfi_+(c)&:=\max_{\alpha\in[0,1]}\{c\alpha-\psi(\alpha)\}=\vfi(c)+c,\ds\ds\ds c\in\bR,
\label{phi+ def}\\
H_r(\rho\|\sigma)
&:=\sup_{\alpha\in(0,1)}\frac{\alpha-1}{\alpha}\left[r-D_{\alpha}(\rho\|\sigma)\right] \nonumber\\
&=\sup_{\alpha\in(0,1)}\frac{(\alpha-1)r-\psi(\alpha)}{\alpha} \nonumber\\
&=\sup_{u\in(-\infty,0)}\{ur-\tilde\psi(u)\},\ds\ds\ds r\in\bR,
\label{Hoeffding def}
\end{align}
where $\tilde\psi(u):=(1-u)\psi((1-u)\inv)$, and $H_r(\rho\|\sigma)$
is the \ki{Hoeffding divergence} of $\rho$ and $\sigma$ with parameter 
$r$. 
In the problem of asymptotic binary state discrimination with null hypothesis $\rho$ and alternative hypothesis $\sigma$, 
$H_r(\rho\|\sigma)$ gives the optimal type I error exponent when the type II exponent 
is at least $r$ \cite{Hayashicq,Nagaoka}. The functions 
$\vfi$ and $\vfi_+$ give a different description of the trade-off curve of the 
two exponents; see, e.g., \cite{Nagaoka}.

\begin{lemma}\label{lemma:phi properties}
Let $\rho,\sigma\in\S(\hil)$, and $\vfi,\vfi_+$ be as above.
\begin{enumerate}
\item
$\vfi_+$ is constant $-\psi(0)$ on $(-\infty,\psi'(0)]$, and it is strictly increasing on 
$[\psi'(0),+\infty)$. 
\item
$\vfi$ is strictly decreasing on $(-\infty,\psi'(1)]$, and 
it is constant $-\psi(1)$ on $[\psi'(1),+\infty)$.
\item
For every $r\ge -\psi(0)$ there exists a unique $c_r\ge\psi'(0)$ such that 
\begin{align}\label{Hr representation}
\vfi_+(c_r)=r,\ds\ds\ds \vfi(c_r)=H_r(\rho\|\sigma).
\end{align}
\item
$r\mapsto H_r(\rho\|\sigma)$ is convex, lower semi-continuous, and monotone decreasing on $\bR$, and 
\begin{align}
r<-\psi(0)=D_0(\rho\|\sigma)&\ds\iff\ds H_r(\rho\|\sigma)=+\infty,\label{Hr infty}\\
r< D(\rho\|\sigma)&\ds\iff\ds H_r(\rho\|\sigma)>0.\label{Hr null}
\end{align}
\end{enumerate}
\end{lemma}
\begin{proof}
The first two points are straightforward to verify (see also \cite[Lemma 4.1]{HMO2}).
For the third, see \cite[Section 2]{Nagaoka} (or the proof of \cite[Theorem 4.8]{HMO2},
with the role of $\rho$ and $\sigma$ interchanged). 
The properties of $r\mapsto H_r(\rho\|\sigma)$ listed in the fourth point are straightforward to verify;
the equivalences in \eqref{Hr infty}--\eqref{Hr null} follow immediately from 
the monotonicity of $D_{\alpha}(\rho\|\sigma)$ in $\alpha$
(see Lemma \ref{lemma:Dalpha monotone}), with the only exception of
the case 
$r=-\psi(0)$, for which a direct calculation yields 
\begin{align}\label{Hr at r=D_0}
H_{-\psi(0)}(\rho\|\sigma)=-\psi'(0)-\psi(0)<+\infty.
\end{align}
\end{proof}

\begin{lemma}
Let $\rho,\sigma\in\S(\hil)$.
For any $b\in\bR$,
\begin{align}\label{Nagaoka bound}
-\vfi(b)=\lim_{n\to+\infty}\frac{1}{n}\log
\min_{T_n\in\bT(\hil^{\otimes n})}\bz\Tr\rho^{\otimes n}(I-T_n)+e^{nb}
\Tr\sigma^{\otimes n}T_n\jz.
\end{align}
\end{lemma}
\begin{proof}
See \cite[Corollary 3.3]{HMO2}.
\end{proof}

The following is a slight variation of the well-known \ki{quantum Hoeffding bound theorem}
\cite{ANSzV,Hayashicq,Hayashibook2,Nagaoka}:
\begin{lemma}\label{lemma:Hoeffding}
Let $\rho,\sigma\in\S(\hil)$.
For any $r\in(0,+\infty)$ and $\alpha\in(0,1)$, 
\begin{align}\label{Hoeffding attainability}
\Tr\sigma^{\otimes n}T_{n,r,\alpha}\le e^{-nr},\ds\ds\ds
\Tr\rho^{\otimes n}(I-T_{n,r,\alpha})\le e^{-n\frac{\alpha-1}{\alpha}\left[r-D_{\alpha}(\rho\|\sigma)\right]},
\end{align}
where $T_{n,r,\alpha}$ is the spectral projection 
of $\rho^{\otimes n}-e^{n(r+\psi(\alpha))/\alpha}\sigma^{\otimes n}$
corresponding to its positive eigenvalues.

Conversely, for any $r>D_0(\rho\|\sigma)$, any test sequence $T_n\in\bT(\hil^{\otimes n})$, $n\in\bN$, and any strictly increasing sequence $(n_k)_{k\in\bN}$ in $\bN$,
\begin{align}\label{Hoeffding converse}
\text{if}\ds\ds
\liminf_{k\to+\infty}-\frac{1}{n_k}\log\Tr\sigma^{\otimes n_k}T_{n_k}\ge r
\ds\ds\text{then}\ds\ds
\limsup_{k\to+\infty}-\frac{1}{n_k}\log\Tr\rho^{\otimes n_k}(I-T_{n_k})\le H_r(\rho\|\sigma).
\end{align}
\end{lemma}
\begin{proof}
The inequalities in \eqref{Hoeffding attainability} 
follow immediately by the application of the trace inequality 
$\half\Tr(A+B)-\half\Tr|A-B|\le\Tr A^{\alpha}B^{1-\alpha}$ in \cite{Aud}
to $A:=\rho^{\otimes n}$ and $B:=e^{n(r+\psi(\alpha))/\alpha}\sigma^{\otimes n}$; 
see \cite[Sec.~3.7]{Hayashibook2} for details.

For the proof of \eqref{Hoeffding converse}, let $c_r$ be as in 
\eqref{Hr representation}, 
and let $b\in(\psi'(0),c_r)$.
By \eqref{Nagaoka bound} we have 
\begin{align*}
-\vfi(b)
&\le
\liminf_{k\to+\infty}\frac{1}{n_k}\log\bz\Tr\rho^{\otimes n_k}(I- T_{n_k})+e^{n_kb}
\Tr\sigma^{\otimes n_k} T_{n_k}\jz\\
&\le
\max\Big\{
\liminf_{k\to+\infty}\frac{1}{n_k}\log\Tr\rho^{\otimes n_k}(I- T_{n_k}),
b+\underbrace{\limsup_{k\to+\infty}\frac{1}{n_k}\log\Tr\sigma^{\otimes n_k} T_{n_k}}_{
\le-r}\Big\},
\end{align*}
where the second inequality follows since, by their very definitions,
$\liminf_k{1\over c_k}\log(x_k+y_k)\le\max\{\liminf_k{1\over c_k}\log x_k,
\limsup_k{1\over c_k}\log y_k\}$ for non-negative real sequences 
$(x_k)_{k\in\bN}$, $(y_k)_{k\in\bN}$,
and $0<c_k\to+\infty$.
Since $b-r-(-\vfi(b))=\vfi_+(b)-\vfi_+(c_r)<0$, 
according to Lemma \ref{lemma:phi properties}, 
we get 
$-\vfi(b)\le \liminf_{k\to+\infty}\frac{1}{n_k}\log\Tr\rho^{\otimes n}(I- T_{n_k})$.
Taking $b\nearrow c_r$, we get 
\begin{align*}
-H_r(\rho\|\sigma)=-\vfi(c_r)\le \liminf_{k\to+\infty}\frac{1}{n_k}\log\Tr\rho^{\otimes n_k}(I- T_{n_k}),
\end{align*}
as required.
\end{proof}

\begin{rem}
\eqref{Hoeffding converse} 
is essentially the same as the converse part of the Hoeffding bound theorem given in 
\cite{Nagaoka}, the only difference being the restriction to a subsequence, 
which is the form we will need it in the proof of Theorem \ref{thm:testdiv}. 
The proof above is exactly the same as the one in \cite{Nagaoka};
we give it in detail since the subsequence version does not follow formally from 
the corresponding statement in \cite{Nagaoka}.
\end{rem}

\section{Test-measured R\'enyi divergences in finite dimension}

\subsection{Definitions}
\label{sec:def}

For $\rho,\sigma\in\S(\hil)$ and $\alpha\in(0,+\infty)$, 
their \ki{measured R\'enyi $\alpha$-divergence} is defined as
\begin{align*}
D_{\alpha}^{\meas}(\rho\|\sigma):=
\sup\left\{D_{\alpha}(\M(\rho)\|\M(\sigma)):\,M\in\povm(\hil,[d]),\s d\in\bN\right\}.
\end{align*}
It is known \cite{BFT_variational,HiaiMosonyi2017} that 
\begin{align}\label{meas def}
D_{\alpha}^{\meas}(\rho\|\sigma)&=
\max\left\{D_{\alpha}(\M(\rho)\|\M(\sigma)):\,M\in\pvm_1(\hil)\right\}.
\end{align}
By restricting to $2$-outcome POVMs, we get the notion of 
the \ki{test-measured R\'enyi $\alpha$-divergence} 
of $\rho$ and $\sigma$,
\begin{align}
D_{\alpha}^{\test}(\rho\|\sigma)
&:=
\sup\{D_{\alpha}(\M(\rho)\|\M(\sigma)):\,M\in\povm(\hil,[2])\}\nn
\\
&=
\max_{T\in\bT(\hil)}D_{\alpha}\bz\T(\rho)\|\T(\sigma)\jz,\label{test-div def2}
\end{align}
where for $T\in\bT(\hil)$ we use the notation
\begin{align}\label{postmeas}
\T(X):=(\Tr XT,\Tr X(I-T))\in[0,+\infty)^2,\ds\ds\ds X\in\B(\hil).
\end{align}
It is obvious from the definitions that
\begin{align}\label{test-meas ineq}
D_{\alpha}^{\test}(\rho\|\sigma)\le D_{\alpha}^{\meas}(\rho\|\sigma).
\end{align}

It is easy to see that the maximum in \eqref{test-div def2} exists. 
Indeed, when $\alpha\in(0,1)$, it follows from the (joint) continuity of the classical 
R\'enyi $\alpha$-divergence in its arguments.  
When $\alpha>1$ and $\rho^0\not\le\sigma^0$ then $T:=\sigma^0$ yields 
$D_{\alpha}\bz\T(\rho)\|\T(\sigma)\jz=+\infty=D_{\alpha}^{\test}(\rho\|\sigma)$. Finally,
when $\alpha>1$ and $\rho^0\le\sigma^0$ then $\rho\le\lambda\sigma$ for some $\lambda>0$, and 
in this case $D_{\alpha}\bz\T(\rho)\|\T(\sigma)\jz$ is continuous in $T$; see, e.g., 
\cite[Remark 4.16]{HiaiMosonyi2017} for details.

Note that every $2$-outcome POVM can be decomposed into a convex combination
of projective $2$-outcome measurements; this follows from the Krein-Milman theorem and the fact that the extreme points of the set of tests are exactly the projections
(see, e.g., \cite[Chap.~I, Lemma 10.1]{Takesaki1}, 
or \cite[Page 23]{HolevoStatistical} for an alternative argument.)
Since $D_{\alpha}$ is jointly quasi-convex in its arguments, we get that 
\begin{align}
D_{\alpha}^{\test}(\rho\|\sigma)
&=
\max\{D_{\alpha}(\M(\rho)\|\M(\sigma)):\,\M\in\pvm(\hil,[2])\}\label{test-div pr1}\\
&=
\max_{T\in\bP(\hil)}D_{\alpha}\bz\T(\rho)\|\T(\sigma)\jz.\label{test-div pr2}
\end{align}

\begin{rem}
For $\alpha=1$, $D^{\meas}:=D^{\meas}_1$ and $D^{\test}:=D^{\test}_1$
are also called the \ki{measured relative entropy} and the 
\ki{test-measured relative entropy}, respectively.
\end{rem}

\begin{rem}\label{rem:subalgebra}
If $\rho$ and $\sigma$ are in a unital $^*$-subalgebra $\A\subseteq\B(\hil)$
with unit $I_{\hil}$, and 
$\E$ is the trace-preserving conditional expectation (equivalently, the orthogonal projection 
with respect to the Hilbert-Schmidt inner product) onto $\A$, then the simple identity
$\Tr X T=\Tr\E(X)T=\Tr X\E(T)$, $X\in\A$, 
implies that an optimal $T$ attaining the maximum in \eqref{test-div def2}
exists with $T\in\A$. 
Moreover, by the same argument leading to \eqref{test-div pr1}--\eqref{test-div pr2}, there exists a projection $T\in\A$ attaining the maximum in \eqref{test-div pr2}.
In particular, if $\rho$ and $\sigma$ commute, and hence they are diagonal in a common orthonormal basis, then there exists an optimal $T$ in the sense of \eqref{test-div pr2}
that is also diagonal in the same basis. 
\end{rem}

\begin{rem}
It is clear from \eqref{meas def} and Remark \ref{rem:subalgebra} that if the subalgebra generated by $\rho$ and $\sigma$ is isomorphic to a $^*$-subalgebra of $\B(\bC^2)$ then 
$D_{\alpha}^{\test}(\rho\|\sigma)=D_{\alpha}^{\meas}(\rho\|\sigma)$ for every 
$\alpha\in(0,+\infty)$.
\end{rem}

\begin{lemma}\label{lemma:testdiv monotone in alpha}
For any $\rho,\sigma\in\S(\hil)$,
$\alpha\mapsto D_{\alpha}^{\test}(\rho\|\sigma)$ is monotone increasing on 
$(0,+\infty)$. 
\end{lemma}
\begin{proof}
Obvious from the monotonicity of the classical R\'enyi $\alpha$-divergences in $\alpha$; see, e.g., 
Lemma \ref{lemma:Dalpha monotone}.
\end{proof}

\begin{lemma}\label{lemma:testdiv strictly positive}
$D_{\alpha}^{\test}$ is strictly positive in the sense that for any $\rho,\sigma\in\S(\hil)$, 
\begin{align}\label{testdiv nonneg}
D_{\alpha}^{\test}(\rho\|\sigma)\ge 0,
\end{align}
with equality if and only if $\rho=\sigma$.
\end{lemma}
\begin{proof}
By the monotonicity stated in Lemma \ref{lemma:testdiv monotone in alpha}, it is sufficient to 
prove strict positivity of $D_{\alpha}^{\test}$ for $\alpha\in(0,1)$.
Non-negativity in \eqref{testdiv nonneg} is 
obvious, since $D_{\alpha}(\T(\rho)\|\T(\sigma))=0$ for $T=I$.
The implication $\rho=\sigma\imp D_{\alpha}^{\test}(\rho\|\sigma)= 0$
is again obvious.
Conversely, assume that $D_{\alpha}^{\test}(\rho\|\sigma)= 0$. Then 
$D_{\alpha}\bz\T(\rho)\|\T(\sigma)\jz=0$ for $T=\pr{\psi}$, where $\psi$ can be any unit vector. Using the strict positivity of the classical R\'enyi $\alpha$-divergence, 
(which is a simple consequence of H\"older's inequality),
we get that 
$\Tr\rho\pr{\psi}=\Tr\sigma\pr{\psi}$, $\psi\in\hil$, whence $\rho=\sigma$.
\end{proof}

\begin{rem}
It is obvious from the definitions that for any $\alpha\in(0,+\infty)$, 
both $D_{\alpha}^{\meas}$ and 
$D_{\alpha}^{\test}$ are monotone non-increasing under the application 
of the same positive trace-preserving map on both of their arguments.
\end{rem}

The \ki{regularized measured R\'enyi $\alpha$-divergence} of $\rho$ and $\sigma$ is defined as
\begin{align*}
\oll{D}_{\alpha}^{\meas}(\rho\|\sigma):=
\sup_{n\in\bN}\frac{1}{n}D_{\alpha}^{\meas}(\rho^{\otimes n}\|\sigma^{\otimes n})
=
\lim_{n\to+\infty}\frac{1}{n}D_{\alpha}^{\meas}(\rho^{\otimes n}\|\sigma^{\otimes n}),
\end{align*}
where the equality follows from the easily verifiable super-additivity of 
$n\mapsto D_{\alpha}^{\meas}(\rho^{\otimes n}\|\sigma^{\otimes n})$ due to Fekete's lemma.

For the \ki{test-measured R\'enyi $\alpha$-divergence}, we consider two potentially different regularizations: 
\begin{align*}
\hat D_{\alpha}^{\test}(\rho\|\sigma):=
\sup_{n\in\bN}\frac{1}{n}D_{\alpha}^{\test}(\rho^{\otimes n}\|\sigma^{\otimes n}),
\end{align*}
and
\begin{align}
\oll{D}_{\alpha}^{\test}(\rho\|\sigma)
&:=
\limsup_{n\to+\infty}\frac{1}{n}
D_{\alpha}^{\test}\bz\rho^{\otimes n}\|\sigma^{\otimes n}\jz\nn
\\
&=
\sup_{(T_n)_{n\in\bN}}\left\{\limsup_{n\to+\infty}\frac{1}{n}
D_{\alpha}\bz\T_n(\rho^{\otimes n})\|\T_n(\sigma^{\otimes n})\jz
\right\}\nn\\
&=
\max_{(T_n)_{n\in\bN}}\left\{\limsup_{n\to+\infty}\frac{1}{n}
D_{\alpha}\bz\T_n(\rho^{\otimes n})\|\T_n(\sigma^{\otimes n})\jz
\right\},\label{testdiv def}
\end{align}
where the optimizations are taken over all sequences of tests
$T_n\in\bT(\hil^{\otimes n})$, $n\in\bN$.

Note that, unlike for the measured R\'enyi divergence, it is not obvious from the 
definition whether 
$n\mapsto D_{\alpha}^{\test}(\rho^{\otimes n}\|\sigma^{\otimes n})$ 
is super-additive, and neither is it obvious 
whether the above two notions of regularized
test-measured R\'enyi divergence coincide. 
It is a non-trivial fact that they do for $\alpha\ge 1$, and in fact, more is true:
\begin{align}\label{testdiv equalities}
\oll{D}_{\alpha}^{\test}
=
\hat D_{\alpha}^{\test}
=
\oll{D}_{\alpha}^{\meas}
=
D_{\alpha}\nw,\ds\ds\ds \alpha\ge 1,
\end{align}
as was shown in \cite{HP} for $\alpha=1$, and 
in \cite{MO} for $\alpha>1$.
It might be natural to conjecture that the equalities 
\begin{align*}
\oll{D}_{\alpha}^{\test}
=
\hat D_{\alpha}^{\test}
=
\oll{D}_{\alpha}^{\meas}
\end{align*}
hold also for $\alpha\in(0,1)$. 
(For the relation of $\oll{D}_{\alpha}^{\meas}$ and $D_{\alpha}\nw$ for $\alpha\in(0,1)$,
see Lemma \ref{lemma:Renyi order} below.)
However, we show that this is not the case, as for $\alpha\in(0,1)$ 
we have 
\begin{align*}
\oll{D}_{\alpha}^{\test}(\rho\|\sigma)
<D_{\alpha}(\rho\|\sigma)
\end{align*}
in general, according to Theorem \ref{thm:testdiv standard bounds} below, and similarly, we have 
\begin{align*}
\hat D_{\alpha}^{\test}(\rho\|\sigma)
<
D_{\alpha}(\rho\|\sigma)
\end{align*}
for any two unequal commuting states with equal supports, as we show in Theorem \ref{thm:an2 main}.
Since for commuting states $D_{\alpha}(\rho\|\sigma)=\oll{D}_{\alpha}^{\meas}(\rho\|\sigma)$,
this yields that neither 
$\oll{D}_{\alpha}^{\test}$ nor 
$\hat D_{\alpha}^{\test}$ is equal to 
$\oll{D}_{\alpha}^{\meas}$ for $\alpha\in(0,1)$.
In particular, we obtain the following:
\begin{cor}
For $\alpha\in(0,1)$, $\oll{D}_{\alpha}^{\test}$ and $\hat D_{\alpha}^{\test}$ are 
not quantum R\'enyi $\alpha$-divergences in the sense of Definition \ref{def:qRenyi}.
\end{cor}

Moreover, we show in Theorem \ref{thm:different reg} that for certain commuting states
\begin{align*}
\oll{D}_{\alpha}^{\test}(\rho\|\sigma)<\hat D^{\test}_{\alpha}(\rho\|\sigma)
\end{align*}
holds. Clearly, for these states
\begin{align*}
n\mapsto D_{\alpha}^{\test}(\rho^{\otimes n}\|\sigma^{\otimes n})\ds\ds\text{is not superadditive}. 
\end{align*}

\begin{rem}
One may argue that for the study of the strong converse exponent of asymptotic binary state
discrimination, the conceptually most natural quantum R\'enyi $\alpha$-divergence is 
$\oll{D}_{\alpha}^{\test}$. To see this, note that for any sequence of tests
$(T_n)_{n\in\bN}$, and any $\alpha>1$,
\begin{align*}
D_{\alpha}^{\test}(\rho^{\otimes n}\|\sigma^{\otimes n})
&\ge
\frac{1}{\alpha-1}\log\bz
(\Tr\rho^{\otimes n}T_n)^{\alpha}(\Tr\sigma^{\otimes n}T_n)^{1-\alpha}
+
(1-\Tr\rho^{\otimes n}T_n)^{\alpha}(1-\Tr\sigma^{\otimes n}T_n)^{1-\alpha}\jz\\
&\ge
\frac{\alpha}{\alpha-1}\log\Tr\rho^{\otimes n}T_n
-
\log \Tr\sigma^{\otimes n}T_n,
\end{align*}
where the first inequality is by definition, and the second one is trivial. This yields that 
\begin{align}
\scli_r(\rho\|\sigma)&:=
\inf\left\{\liminf_{n\to+\infty}-\frac{1}{n}\log\Tr\rho^{\otimes n}T_n:\,
T_n\in\bT(\hil^{\otimes n}),\,n\in\bN,\s\liminf_{n\to+\infty}-\frac{1}{n}\log\Tr\sigma^{\otimes n}T_n\ge r
\right\}\nn\\
&\ge
\sup_{\alpha>1}\frac{\alpha-1}{\alpha}\left[r-\oll{D}_{\alpha}^{\test}(\rho\|\sigma)\right],
\label{sc optimality}
\end{align}
which is the optimality part of the theorem for the strong converse exponent. 
The above argument is essentially the by now standard argument given in  
\cite{N}, except that instead of using the monotonicity of some family of quantum R\'enyi 
$\alpha$-divergences as in \cite{N} and \cite{MO}, the inequalities follow immediately by 
definition. The attainability part of the strong converse theorem given in \cite{MO} tells that 
equality holds in \eqref{sc optimality}; moreover, we have 
$\oll{D}_{\alpha}^{\test}(\rho\|\sigma)=D_{\alpha}\nw(\rho\|\sigma)$, $\alpha>1$.
\end{rem}
\bigskip

In the rest of this section we discuss known relations between the different quantum 
R\'enyi divergences introduced above, and also some new inequalities.

We will often benefit from the following simple observation:

\begin{lemma}\label{lemma:alpha swap2}
For any $\rho,\sigma\in\S(\hil)$, 
\begin{align}\label{alpha swap}
(1-\alpha)D_{\alpha}(\rho\|\sigma)=\alpha D_{1-\alpha}(\sigma\|\rho),\ds\ds\ds
\alpha\in(0,1).
\end{align}
If $\divv_{\alpha}(\rho\|\sigma)$ denotes any of 
$D_{\alpha}^{\meas}(\rho\|\sigma)$,
$\oll{D}_{\alpha}^{\meas}(\rho\|\sigma)$,
$D_{\alpha}^{\test}(\rho\|\sigma)$,
$\oll{D}_{\alpha}^{\test}(\rho\|\sigma)$,
$\hat D_{\alpha}^{\test}(\rho\|\sigma)$, then 
\begin{align}\label{alpha swap2}
(1-\alpha)\divv_{\alpha}(\rho\|\sigma)=\alpha \divv_{1-\alpha}(\sigma\|\rho),\ds\ds\ds\alpha\in(0,1).
\end{align}
\end{lemma}
\begin{proof}
The identity in \eqref{alpha swap} 
is well-known, and follows immediately from the definition in \eqref{standard Renyi def}.
In particular, this identity holds for the classical R\'enyi divergence, which yields
\eqref{alpha swap2}.
\end{proof}

\begin{lemma}\label{lemma:Renyi order}
For any $\rho,\sigma\in\S(\hil)$ and $\alpha\in(0,+\infty)$,

\begin{align}
&\ds
\ds\ds\ds\ds\ds\ds\ds\ds\ds D_{\alpha}^{\test}(\rho\|\sigma) \s\le\s D_{\alpha}^{\meas}(\rho\|\sigma)
\label{Renyi order0}\\
&\ds\ds\ds\ds\ds\ds\ds\ds\ds\ds\ds\ds\vertleq\ds\ds\ds\ds\ds\ds\ds\ds\ds\vertleq\nn\\
&\oll{D}_{\alpha}^{\test}(\rho\|\sigma)
\s\le\s
\hat D_{\alpha}^{\test}(\rho\|\sigma)
\s\le\s
\oll{D}_{\alpha}^{\meas}(\rho\|\sigma)\label{Renyi order1}\\
&\ds\ds\ds\ds\ds\ds\ds\ds\ds\ds\ds\ds\ds\ds\ds\ds\ds\ds=
\begin{cases}
D_{\alpha}\nw(\rho\|\sigma),&\alpha\in[1/2,+\infty),\\
\frac{\alpha}{1-\alpha}D_{1-\alpha}\nw(\sigma\|\rho),&\alpha\in(0,1/2]
\end{cases}
\label{Renyi order3}\\
&\ds\ds\ds\ds\ds\ds\ds\ds\ds\ds\ds\ds\ds\ds\ds\ds\ds\ds\le
D_{\alpha}(\rho\|\sigma).\label{Renyi order4}
\end{align}
\end{lemma}
\begin{proof}
The inequalities in \eqref{Renyi order0} and \eqref{Renyi order1} are obvious by definition. 
The equality in \eqref{Renyi order3} was given in \cite{HP} for $\alpha=1$, 
in \cite{MO} for $\alpha>1$, in \cite{HT14} for $\alpha\in[1/2,1)$, and the case 
$\alpha\in(0,1/2)$ follows from the latter due to \eqref{alpha swap} and \eqref{alpha swap2}.
The inequality in \eqref{Renyi order4} follows from the Araki-Lieb-Thirring inequality, as 
observed in \cite{Renyi_new,WWY}
(where the case $\alpha\in[1/2,+\infty)$ was treated; the case 
$\alpha\in(0,1/2)$ follows again
from this due to \eqref{alpha swap} and \eqref{alpha swap2}).
\end{proof}

\begin{rem}\label{rem:strict positivity}
Lemmas \ref{lemma:Renyi order} and \ref{lemma:testdiv strictly positive} imply that 
all the quantities dominating $D_{\alpha}^{\test}(\rho\|\sigma)$ in 
\eqref{Renyi order0}--\eqref{Renyi order4} are 
non-negative, and they are strictly positive when $\rho\ne\sigma$.
This does not imply the strict positivity of
$\oll{D}_{\alpha}^{\test}$; we will prove that in Corollary \ref{cor:regularized strictly positive}.
\end{rem}

\begin{rem}\label{rem:Renyi order strict}
In the chain of inequalities
\begin{align}\label{strict inequalities1}
D_{\alpha}^{\meas}(\rho\|\sigma)
\le
\oll{D}_{\alpha}^{\meas}(\rho\|\sigma)
\le
D_{\alpha}(\rho\|\sigma)
\end{align}
above, both inequalities are equalities when $\rho$ and $\sigma$ commute, i.e., 
\begin{align*}
\rho\sigma=\sigma\rho\ds\ds\imp\ds\ds
D_{\alpha}^{\meas}(\rho\|\sigma)
=
\oll{D}_{\alpha}^{\meas}(\rho\|\sigma)
=
D_{\alpha}(\rho\|\sigma),\ds\ds\alpha\in(0,+\infty),
\end{align*}
as one can easily verify.
Hence, in this case we will only use the notation 
$D_{\alpha}(\rho\|\sigma)$ to denote the unique quantum R\'enyi $\alpha$-divergence
of $\rho$ and $\sigma$.
Note that it also coincides with $D_{\alpha}\nw(\rho\|\sigma)$.

On the other hand, the second inequality in \eqref{strict inequalities1} is strict whenever
$\rho$ and $\sigma$ are non-commuting states
and $D_{\alpha}(\rho\|\sigma)<+\infty$;
this follows from 
\cite[Proposition 2.1]{Hiai-ALT}, taking also into account 
\eqref{alpha swap} and \eqref{alpha swap2} when $\alpha\in(0,1/2)$.
\end{rem}
\medskip

Strict inequality in the first inequality in \eqref{strict inequalities1} was proved in 
\cite[Theorem 6]{BFT_variational} for $\alpha\in(1/2,+\infty)$ 
and non-commuting invertible states. In the next proposition  
we give a slightly stronger statement, with a proof that is different from the one in \cite{BFT_variational}. 

\begin{prop}\label{prop:single-copy equality meas}
For any $\rho,\sigma\in\S(\hil)$,
\begin{align}
D_{\alpha}^{\meas}(\rho\|\sigma)\le
\begin{cases}
\oll{D}_{\alpha}^{\meas}(\rho\|\sigma),&\alpha\in(0,1/2)\cup(1/2,1),\\
D_{\alpha}(\rho\|\sigma),&\alpha=1/2.
\end{cases}
\label{classical sh strict ineq1 meas}
\end{align}
Moreover, if one of the following holds, then $\rho$ and $\sigma$ commute:
\begin{enumerate}
\item $\rho^0\le\sigma^0$ and
$D_{\alpha}^{\meas}(\rho\|\sigma)=\oll{D}_{\alpha}^{\meas}(\rho\|\sigma)$ for some $\alpha\in(1/2,1)$;
\item $\rho^0\ge\sigma^0$ and
$D_{\alpha}^{\meas}(\rho\|\sigma)=\oll{D}_{\alpha}^{\meas}(\rho\|\sigma)$ for some $\alpha\in(0,1/2)$;
\item $\rho^0\le\sigma^0$ or $\rho^0\ge\sigma^0$, and
$D_{1/2}^{\meas}(\rho\|\sigma)=D_{1/2}(\rho\|\sigma)$.
\end{enumerate}
\end{prop}
\begin{proof}
The inequality in 
\eqref{classical sh strict ineq1 meas} is a special case of 
the inequalities given in Lemma \ref{lemma:Renyi order}.

Assume that $\rho^0\le\sigma^0$, and 
\begin{align*}
&\alpha=1/2\ds\ds\ds\s\ds\text{and}\ds\ds
D_{\alpha}^{\meas}(\rho\|\sigma)
=
D_{\alpha}(\rho\|\sigma),\ds\ds\text{or}\\
&\alpha\in(1/2,1)\ds\ds\text{and}\ds\ds
D_{\alpha}^{\meas}(\rho\|\sigma)
=
\oll{D}_{\alpha}^{\meas}(\rho\|\sigma)=
D_{\alpha}\nw(\rho\|\sigma),
\end{align*}
where the last equality is due to \cite{HT14}.
According to \eqref{meas def}, in each case above there exists a measurement
$M_{\alpha}\in\pvm_1(\hil)$ such that 
with the CPTP map
\begin{align*}
\M_{\alpha}(X):=\sum_{i=1}^{\dim\hil}(\Tr M_{\alpha}X)\pr{i},\ds\ds X\in\B(\hil),
\end{align*}
where $(\ket{i})_{i=1}^{\dim\hil}$ is an orthonormal basis in $\bC^{\dim\hil}$,
we have 
\begin{align*}
D_{\alpha}\bz\M_{\alpha}(\rho)\|\M_{\alpha}(\sigma)\jz
=
\begin{cases}
D_{\alpha}(\rho\|\sigma),&\alpha=1/2,\\
D_{\alpha}\nw(\rho\|\sigma),&\alpha\in(1/2,1).
\end{cases}
\end{align*}
In the case $\alpha\in(1/2,1)$, 
the results of \cite{Jencova_NCLpII} combined with the above imply that
\begin{align*}
\bz\M_{\alpha}\jz_{\sigma}^*(\M_{\alpha}(\rho))=\rho,\ds\ds
\bz\M_{\alpha}\jz_{\sigma}^*(\M_{\alpha}(\sigma))=\sigma,
\end{align*}
where 
\begin{align*}
\bz\M_{\alpha}\jz_{\sigma}^*(Y):=
\sigma^{1/2}
\M_{\alpha}^*\bz (\M_{\alpha}(\sigma))^{-1/2}Y(\M_{\alpha}(\sigma))^{-1/2}\jz
\sigma^{1/2},\ds\ds\ds Y\in\B(\bC^{\dim\hil}),
\end{align*}
is a CPTP map, called the \ki{Petz map} \cite{Petz1988}.
A double application of the monotonicity of $D_{1/2}$ under CPTP maps \cite{P86} then yields
\begin{align*}
D_{1/2}\bz \M_{\alpha}(\rho)\|\M_{\alpha}(\sigma)\jz
=
D_{1/2}\bz\rho\|\sigma\jz.
\end{align*}
Using a suitable modification of the proof of \cite[Lemma 4.1]{Petz2003}, where 
the invertibility of the states is replaced with the condition 
$D_1^0\le D_2^0$ (equivalent to our assumption that 
$\rho^0\le\sigma^0$), one then obtains that 
$\rho$ and $\sigma$ commute.
(See also \cite[Sec.~7.2]{Hiai_fdiv_Springer} in a more general setting.)

This completes the proof of the assertion about the equality 
in \eqref{classical sh strict ineq1 meas}
in the case $\alpha\in[1/2,1)$.
The case $\alpha\in(0,1/2]$ follows from this immediately due to 
\eqref{alpha swap} and 
the identities in 
Lemma \ref{lemma:alpha swap2}.
\end{proof}

\begin{rem}
For $\alpha=1/2$ we have $D_{1/2}^{\meas}(\rho\|\sigma)=-2\log F(\rho,\sigma)=
D_{1/2}\nw(\rho\|\sigma)=\oll{D}_{1/2}^{\meas}(\rho\|\sigma)$, where 
$F(\rho,\sigma):=\Tr\bz\rho^{1/2}\sigma\rho^{1/2}\jz^{1/2}$ is the \ki{fidelity} of $\rho$ and 
$\sigma$, the first equality is explained, e.g., in \cite[Sec.~9]{NC}, the second equality is by 
definition, and the last equality 
(which is a special case of the equality in \eqref{Renyi order3})
follows from the above and the monotonicity 
of the fidelity under CPTP maps and its 
multiplicativity under tensor products.
In particular, $D_{1/2}^{\meas}(\rho\|\sigma)=\oll{D}_{1/2}^{\meas}(\rho\|\sigma)$ always holds, 
and does not imply the commutativity of $\rho$ and $\sigma$, which is why the 
$\alpha=1/2$ case is treated separately in Proposition \ref{prop:single-copy equality meas}.
\end{rem}

\begin{ex}\label{prop:pure states}
Let $\psi,\phi\in\hil$ be unit vectors that are neither parallel nor perpendicular. 
For any $\alpha\in(0,1)$,
\begin{align}
D_{\alpha}\bz\pr{\psi}\|\pr{\phi}\jz&=\frac{1}{\alpha-1}\log|\inner{\psi}{\phi}|^2,
\label{pure standard}\\
D_{\alpha}\nw\bz\pr{\psi}\|\pr{\phi}\jz&=\frac{\alpha}{\alpha-1}\log|\inner{\psi}{\phi}|^2,
\label{pure sandwiched}\\
D_{\alpha}^{\test}\bz\pr{\psi}\|\pr{\phi}\jz&=
D_{\alpha}^{\meas}\bz\pr{\psi}\|\pr{\phi}\jz\nn\\
&=
\oll{D}_{\alpha}^{\meas}\bz\pr{\psi}\|\pr{\phi}\jz=
\oll{D}_{\alpha}^{\test}\bz\pr{\psi}\|\pr{\phi}\jz=
\hat D_{\alpha}^{\test}\bz\pr{\psi}\|\pr{\phi}\jz\nn\\
&=
\begin{cases}
-\log|\inner{\psi}{\phi}|^2,&\alpha\in(0,1/2],\\
\frac{\alpha}{\alpha-1}\log|\inner{\psi}{\phi}|^2,&\alpha\in[1/2,1).
\end{cases}
\label{pure measured}
\end{align}
In particular, with $\divv_{\alpha}$ as in Lemma \ref{lemma:alpha swap2},
\begin{align}
D_{\alpha}\nw\bz\pr{\psi}\|\pr{\phi}\jz
&=
\alpha D_{\alpha}\bz\pr{\psi}\|\pr{\phi}\jz\nn\\
&<
(1-\alpha) D_{\alpha}\bz\pr{\psi}\|\pr{\phi}\jz\nn\\
&=
\divv_{\alpha}\bz\pr{\psi}\|\pr{\phi}\jz\nn\\
&<
D_{\alpha}\bz\pr{\psi}\|\pr{\phi}\jz
\label{pure example}
\end{align}
for every $\alpha\in(0,1/2)$, and
\begin{align}\label{pure example2}
\divv_{\alpha}\bz\pr{\psi}\|\pr{\phi}\jz
=
D_{\alpha}\nw\bz\pr{\psi}\|\pr{\phi}\jz
=
\alpha D_{\alpha}\bz\pr{\psi}\|\pr{\phi}\jz
<
D_{\alpha}\bz\pr{\psi}\|\pr{\phi}\jz
\end{align}
for every $\alpha\in(1/2,1)$.

Indeed, the above are easy to see as follows.
First, it follows by a straightforward computation that for any $\sigma\in\S(\hil)$, 
\begin{align*}
D_{\alpha}\nw\bz\pr{\psi}\|\sigma\jz
=
\frac{\alpha}{\alpha-1}\log\inner{\psi}{\sigma^{\frac{1-\alpha}{\alpha}}\psi},
\ds\ds\ds
D_{\alpha}\bz\pr{\psi}\|\sigma\jz
=
\frac{1}{\alpha-1}\log\inner{\psi}{\sigma^{1-\alpha}\psi},
\end{align*}
from which the equalities in \eqref{pure standard} 
and \eqref{pure sandwiched} follow immediately.

We have 
\begin{align}
-\log|\inner{\psi}{\phi}|^2\le D_{\alpha}^{\test}\bz\pr{\psi}\|\pr{\phi}\jz,
\ds\ds\ds
-\log|\inner{\psi}{\phi}|^2\le \oll{D}_{\alpha}^{\test}\bz\pr{\psi}\|\pr{\phi}\jz,
\label{pure proof}
\end{align}
where the first inequality follows by choosing the test 
$T=\pr{\psi}$, and the second inequality by choosing the test sequence
$T_n=\pr{\psi}^{\otimes n}$, $n\in\bN$.
For $\alpha\in(0,1/2]$ we have 
\begin{align}
\oll{D}_{\alpha}^{\meas}(\rho\|\sigma)=\frac{\alpha}{1-\alpha}D_{1-\alpha}\nw\bz\pr{\phi}\|\pr{\psi}\jz
=
\frac{\alpha}{1-\alpha}\frac{1-\alpha}{-\alpha}\log|\inner{\psi}{\phi}|^2
=
-\log|\inner{\psi}{\phi}|^2,\label{pure proof2}
\end{align}
where the first equality is by \eqref{Renyi order3}, and the
second equality is due to \eqref{pure sandwiched}.
Combining \eqref{pure proof} and \eqref{pure proof2} with the inequalities in 
Lemma \ref{lemma:Renyi order} yields the equalities in 
\eqref{pure measured} for $\alpha\in(0,1/2]$, and the equalities for $\alpha\in[1/2,1)$ follow from this 
due to \eqref{alpha swap2}.
The statements in \eqref{pure example}--\eqref{pure example2} are obvious from 
\eqref{pure standard}--\eqref{pure measured}.
\end{ex}

\begin{rem}
As it was shown in \cite[Theorem 7]{BFT_variational}, for any 
$\alpha\in(0,1/2)$ and any non-commuting invertible states $\rho,\sigma$,
the strict inequality
$D_{\alpha}\nw(\rho\|\sigma)<D_{\alpha}^{\meas}(\rho\|\sigma)$ holds.
In view of Example \ref{prop:pure states}, it is natural to ask whether in this setting we also have 
$D_{\alpha}\nw(\rho\|\sigma)<D_{\alpha}^{\test}(\rho\|\sigma)$,
$D_{\alpha}\nw(\rho\|\sigma)<\oll{D}_{\alpha}^{\test}(\rho\|\sigma)$, or
$D_{\alpha}\nw(\rho\|\sigma)<\hat D_{\alpha}^{\test}(\rho\|\sigma)$.
It is also a question whether such strict inequalities may be obtained without any conditions on the 
supports. Note, for instance, that the strict inequalities in Example 
\ref{prop:pure states} are not covered by the results of 
\cite{BFT_variational}, since the states are not invertible.
\end{rem}

\begin{rem}\label{rem:regularized CPTP mon}
It is obvious from the definitions that 
for any $\alpha\in(0,+\infty)$, 
$\oll{D}_{\alpha}^{\meas}$,
$\oll{D}_{\alpha}^{\test}$, and $\hat D_{\alpha}^{\test}$ are all monotone non-increasing 
under the application of the same completely positive trace-preserving 
(CPTP) map on both of their 
arguments.
\end{rem}

\subsection{Analysis of $\oll{D}_{\alpha}^{\test}(\rho\|\sigma)$}
\label{sec:an1}

According to \eqref{testdiv equalities}, both versions of the regularized test-measured 
R\'enyi divergence are the same and coincide with the sandwiched R\'enyi divergence
for $\alpha>1$, and hence for the rest we focus on the case
$\alpha\in(0,1)$. Our key technical result regarding $\oll{D}_{\alpha}^{\test}(\rho\|\sigma)$ is the following:

\begin{thm}\label{thm:testdiv}
For any $\rho,\sigma\in\S(\hil)$ and any
$\alpha\in(0,1)$,
\begin{align}
\oll{D}_{\alpha}^{\test}(\rho\|\sigma)
&=
\lim_{n\to+\infty}\frac{1}{n}D_{\alpha}^{\test}\bz\rho^{\otimes n}\|\sigma^{\otimes n}\jz\label{testdiv limit}\\
&=
\sup_{r>0}\min\left\{r,\frac{\alpha}{1-\alpha}H_r(\rho\|\sigma)\right\}\\
&=\sup\left\{r\ge 0:\,H_r(\rho\|\sigma)\ge \frac{1-\alpha}{\alpha}r \right\}.
\label{testdiv formula}
\end{align}
\end{thm}
\begin{proof}
If $\rho\perp\sigma$ then all the terms in \eqref{testdiv limit}--\eqref{testdiv formula} are equal to 
$+\infty$ and the assertion holds trivially. Hence, for the rest we assume that 
$\rho\not\perp\sigma$. 
On the other hand, if $\rho=\sigma$ then 
all the terms in \eqref{testdiv limit}--\eqref{testdiv formula} are equal to 
$0$ and again the assertion holds trivially. Hence, for the rest we also assume that 
$\rho\ne\sigma$.

Let $r>0$ and $H_r:=H_r(\rho\|\sigma)$. If $r<D(\rho\|\sigma)$ 
(so that $H_r>0$), let $H\in(0,H_r)$, otherwise let $H:=0$.
By Lemma \ref{lemma:Hoeffding},
there exists a sequence of tests $(T_n)_{n\in\bN}$
such that 
\begin{align*}
\Tr\sigma^{\otimes n}T_n\le e^{-nr},\ds\ds\ds
\Tr\rho^{\otimes n}(I-T_n)\le e^{-nH},\ds\ds\ds n\in\bN.
\end{align*}
Along any such sequence, and for any $\alpha\in(0,1)$ and any $n\in\bN$, 
we have 
\begin{align}
&\frac{1}{n}D_{\alpha}\bz\T_n(\rho^{\otimes n})\|\T_n(\sigma^{\otimes n})\jz\nn\\
&\ds=
\frac{1}{n}
\frac{1}{\alpha-1}\log\Bigg(
\underbrace{(\Tr\rho^{\otimes n}T_n)^{\alpha}}_{\le 1}
\underbrace{(\Tr\sigma^{\otimes n}T_n)^{1-\alpha}}_{\le e^{-nr(1-\alpha)}}
+
\underbrace{(\Tr\rho^{\otimes n}(I-T_n))^{\alpha}}_{\le e^{-nH\alpha}}
\underbrace{(\Tr\sigma^{\otimes n}(I-T_n))^{1-\alpha}}_{\le 1}
\Bigg)\nn\\
&\ds\ge
\frac{1}{\alpha-1}\frac{1}{n}
\log\bz e^{-nr(1-\alpha)}+e^{-nH\alpha}\jz\nn\\
&\ds\ge
\frac{1}{\alpha-1}\max\{-(1-\alpha)r,-\alpha H\}-\frac{1}{n}\frac{\log 2}{1-\alpha}\nn\\
&\ds=
\min\left\{r,\frac{\alpha}{1-\alpha}H\right\}-\frac{1}{n}\frac{\log 2}{1-\alpha}.
\label{testdiv proof1}
\end{align}
In particular,
\begin{align}
\liminf_{n\to+\infty}\frac{1}{n}D_{\alpha}\bz\T_n(\rho^{\otimes n})\|\T_n(\sigma^{\otimes n})\jz
\ge
\min\left\{r,\frac{\alpha}{1-\alpha}H\right\}.\label{testdiv proof0}
\end{align}
Thus, for any $\alpha\in(0,1)$, 
\begin{align}
\oll{D}_{\alpha}^{\test}(\rho\|\sigma)
&=
\limsup_{n\to+\infty}\frac{1}{n}D_{\alpha}^{\test}(\rho^{\otimes n}\|\sigma^{\otimes n})\nn\\
&\ge
\liminf_{n\to+\infty}\frac{1}{n}D_{\alpha}^{\test}(\rho^{\otimes n}\|\sigma^{\otimes n})
\label{testdiv att-1}\\
&\ge
\max\left\{0,\sup_{0<r<D(\rho\|\sigma)}\sup_{0\le H<H_r}\min\left\{r,\frac{\alpha}{1-\alpha}H\right\}\right\}\label{testdiv att-2}\\
&=
\sup_{r>0}\min\left\{r,\frac{\alpha}{1-\alpha}H_r\right\}
=\sup\left\{r\ge 0:\,H_r\ge \frac{1-\alpha}{\alpha}r \right\},\label{testdiv att}
\end{align}
where the first equality is by definition, the first inequality is obvious, the second inequality follows by optimizing 
\eqref{testdiv proof0} over the choices made at the beginning of the proof, and the rest are obvious due to \eqref{Hr null}.
Our aim is to show that the inequalities in 
\eqref{testdiv att-1}--\eqref{testdiv att} are in fact equalities.

First, note that by \eqref{Hr infty} 
we can rewrite \eqref{testdiv att} as
\begin{align}\label{testdiv att2}
\oll{D}_{\alpha}^{\test}(\rho\|\sigma)
\ge
\sup_{r>0}\min\left\{r,\frac{\alpha}{1-\alpha}H_r\right\}
=
\max\left\{D_0(\rho\|\sigma),\sup_{r>D_0(\rho\|\sigma)}\min\left\{r,\frac{\alpha}{1-\alpha}H_r\right\}\right\}.
\end{align}
(Note that the lower bound $\oll{D}_{\alpha}^{\test}(\rho\|\sigma)\ge D_0(\rho\|\sigma)$ can also be obtained from the definition of 
$\oll{D}_{\alpha}^{\test}(\rho\|\sigma)$ by choosing the test sequence $T_n=(\rho^0)^{\otimes n}$, $n\in\bN$.)

For the rest we fix an $\alpha\in(0,1)$.
Let $(T_n)_{n\in\bN}$ be a test sequence attaining the maximum in 
\eqref{testdiv def}, and 
let $(n_k)_{k\in\bN}$ be a strictly increasing sequence in $\bN$ such that 
\begin{align*}
\oll{D}_{\alpha}^{\test}(\rho\|\sigma)=
\lim_{k\to+\infty}\frac{1}{n_k}
D_{\alpha}\bz\T_{n_k}(\rho^{\otimes n_k})\|\T_{n_k}(\sigma^{\otimes n_k})\jz.
\end{align*}
For every $n\in\bN$, we have 
\begin{align*}
&(\Tr\rho^{\otimes n}T_n)^{\alpha}(\Tr\sigma^{\otimes n}T_n)^{1-\alpha}
+
(\Tr\rho^{\otimes n}(I-T_n))^{\alpha}(\Tr\sigma^{\otimes n}(I-T_n))^{1-\alpha}\nonumber\\
&\ds\ge
\min\left\{\Tr\sigma^{\otimes n}T_n,\Tr\sigma^{\otimes n}(I-T_n)\right\}^{1-\alpha}\underbrace{\bz(\Tr\rho^{\otimes n}T_n)^{\alpha}+(\Tr\rho^{\otimes n}(I-T_n))^{\alpha}\jz}_{\ge 1},
\end{align*}
whence
\begin{align*}
\oll{D}_{\alpha}^{\test}(\rho\|\sigma)
&\le
\liminf_{k\to+\infty}
-\frac{1}{n_k}\log\min\left\{\Tr\sigma^{\otimes n_k}T_{n_k},\Tr\sigma^{\otimes n_k}(I-T_{n_k})\right\}.
\end{align*}
Let us define a new test sequence 
$\tilde T_n:=T_n$ if $\Tr \sigma^{\otimes n}T_n\le 1/2$, and 
$\tilde T_n:=I-T_n$ otherwise. Then 
\begin{align}
\oll{D}_{\alpha}^{\test}(\rho\|\sigma)
&=
\lim_{k\to+\infty}\frac{1}{n_k}
D_{\alpha}\bz\tilde\T_{n_k}(\rho^{\otimes n_k})\|\tilde\T_{n_k}(\sigma^{\otimes n_k})\jz
\le
\liminf_{k\to+\infty}-\frac{1}{n_k}\log\Tr\sigma^{\otimes n_k}\tilde T_{n_k}
=:r.\label{asymptotic proof2}
\end{align}

Note that $r<D_0(\rho\|\sigma)$ is impossible due the lower bound in 
\eqref{testdiv att2}. 
If $r=D_0(\rho\|\sigma)$ then 
\eqref{testdiv att2} and \eqref{asymptotic proof2} yield that
the inequality in \eqref{testdiv att2} is an equality, 
and hence the inequalities in \eqref{testdiv att-1}--\eqref{testdiv att-2}
are also equalities, proving the theorem.
Hence, for the rest we assume that 
$r>D_0(\rho\|\sigma)$.
By Lemma \ref{lemma:Hoeffding},
\begin{align}\label{Hconverse application}
\limsup_{k\to+\infty}-\frac{1}{n_k}\log\Tr\rho^{\otimes n_k}(I-\tilde T_{n_k})
\le H_r.
\end{align}
Note that 
\begin{align}
D_{\alpha}\bz\tilde\T_{n_k}(\rho^{\otimes n_k})\|\tilde\T_{n_k}(\sigma^{\otimes n_k})\jz
&\le
\frac{1}{\alpha-1}\log \bz
\bz\Tr\rho^{\otimes n_k}(I-\tilde T_{n_k})\jz^{\alpha}
\bz\Tr\sigma^{\otimes n_k}(I-\tilde T_{n_k})\jz^{1-\alpha}
\jz\nn\\
&=
\frac{\alpha}{\alpha-1}\log\Tr\rho^{\otimes n_k}(I-\tilde T_{n_k})
-\log\Tr\sigma^{\otimes n_k}(I-\tilde T_{n_k}).
\label{asymptotic proof1}
\end{align}
By assumption, $r>0$, and hence
$\lim_{k\to+\infty}\Tr\sigma^{\otimes n_k} (I-\tilde T_{n_k})=1$,
according to \eqref{asymptotic proof2}.
Using also  \eqref{Hconverse application} and \eqref{asymptotic proof1}, we get
\begin{align}
\oll{D}_{\alpha}^{\test}(\rho\|\sigma)
&=
\lim_{k\to+\infty}\frac{1}{n_k}
D_{\alpha}\bz\tilde\T_{n_k}(\rho^{\otimes n_k})\|\tilde\T_{n_k}(\sigma^{\otimes n_k})\jz
\nn\\
&\le
\frac{\alpha}{1-\alpha}
\limsup_{k\to+\infty}-\frac{1}{n_k}\log\Tr\rho^{\otimes n_k}(I-\tilde T_{n_k})
\le \frac{\alpha}{1-\alpha}H_r.\label{asymptotic proof3}
\end{align}

By \eqref{asymptotic proof2} and \eqref{asymptotic proof3}, $\oll{D}_{\alpha}^{\test}(\rho\|\sigma)\le\min\{r,\frac{\alpha}{1-\alpha}H_r\}$,
and thus \eqref{testdiv att2} holds as an equality, whence the inequalities in 
\eqref{testdiv att-1}--\eqref{testdiv att} hold as equalities. 
\end{proof}

\begin{rem}
From \eqref{testdiv formula} and \eqref{testdiv proof1} we get immediately that for any 
$\rho,\sigma\in\S(\hil)$, any $\alpha\in(0,1)$, and any $n\in\bN$,
\begin{align*}
\frac{1}{n}D_{\alpha}^{\test}(\rho^{\otimes n}\|\sigma^{\otimes n})+\frac{1}{n}\frac{\log 2}{1-\alpha}
\ge
\oll{D}_{\alpha}^{\test}(\rho\|\sigma);
\end{align*}
in particular,
\begin{align*}
D_{\alpha}^{\test}(\rho\|\sigma)+\frac{\log 2}{1-\alpha}
\ge
\oll{D}_{\alpha}^{\test}(\rho\|\sigma).
\end{align*}
\end{rem}

\begin{rem}\label{rem:decreasing}
Assume that $\rho\ne \sigma$, so that $D(\rho\|\sigma)>0$, and let $H_r:=H_r(\rho\|\sigma)$.
It is clear from the properties listed in Lemma \ref{lemma:phi properties} that
$H_r/r=+\infty$ for all $r<D_0(\rho\|\sigma)$, and 
$r\mapsto H_r/r$ is a strictly decreasing continuous function on 
$(D_0(\rho\|\sigma),D(\rho\|\sigma))$, with 
\begin{align*}
\lim_{r\nearrow D(\rho\|\sigma)}\frac{H_r}{r}=0,\ds\ds\ds
\lim_{r\searrow D_0(\rho\|\sigma)}\frac{H_r}{r}=
\begin{cases}
\frac{\psi'(0)}{\psi(0)}+1,&D_0(\rho\|\sigma)>0,\\
+\infty,&\text{otherwise}.
\end{cases}
\end{align*}
From this it follows that for any $\alpha\in(0,1)$ such that 
$\frac{1-\alpha}{\alpha}< \lim_{r\searrow D_0(\rho\|\sigma)}\frac{H_r}{r}$ there exists a 
unique $r_{\alpha}\in(D_0(\rho\|\sigma),D(\rho\|\sigma))$
such that $\frac{H_{r_{\alpha}}}{r_{\alpha}}=\frac{1-\alpha}{\alpha}$,
and for this $r_{\alpha}$,
\begin{align}\label{testdiv}
\oll{D}_{\alpha}^{\test}(\rho\|\sigma)
=
r_{\alpha}.
\end{align}

Although \eqref{testdiv} still involves in fact continuum many optimizations 
(one for each $H_r$), and the solution of the non-trivial equation 
$H_r=\frac{1-\alpha}{\alpha}r$, it is still remarkable in the sense that it 
is single-letter, i.e., does not involve a limit. In particular, numerical computations for small dimensions are feasible.
\end{rem}

\begin{rem}
A different expression for $\oll{D}_{\alpha}^{\test}(\rho\|\sigma)$ can be obtained by exploiting a connection with a recently studied hypothesis testing problem in \cite{Salzmann_Datta21}; 
this yields
\begin{align}
\oll{D}_{\alpha}^{\test}(\rho\|\sigma)
=
\alpha\sup_{0< t< 1}\frac{(t-1)D_t(\rho\|\sigma)}{t(2\alpha-1)-\alpha},\ds\ds\ds\alpha\in(0,1).
\label{alt repr4}
\end{align}
We give the details in Appendix \ref{sec:SD}.
\end{rem}

The main result of this section is the following:

\begin{thm}\label{thm:testdiv standard bounds}
Let $\rho,\sigma\in\S(\hil)$. For every $\alpha\in(0,1)$,
\begin{align}\label{testdiv standard bounds}
\half D_{\alpha}(\rho\|\sigma)\le \oll{D}^{\test}_{\alpha}(\rho\|\sigma)\le
D_{\alpha}(\rho\|\sigma).
\end{align}
Moreover, the following are equivalent:
\begin{enumerate}
\item\label{strict ineq1}
The second inequality in \eqref{testdiv standard bounds}
is strict for every $\alpha\in(0,1)$.

\item\label{strict ineq2}
The second inequality in \eqref{testdiv standard bounds}
is strict for some $\alpha\in(0,1)$.

\item\label{strict ineq3}
$\frac{1}{n}D_{\alpha}^{\test}(\rho^{\otimes n}\|\sigma^{\otimes n})
<D_{\alpha}(\rho\|\sigma)$ for some $\alpha\in(0,1)$ and some $n\in\bN$.

\item\label{strict ineq4}
$\beta\mapsto D_{\beta}(\rho\|\sigma)$ 
and
$\beta\mapsto D_{\beta}(\sigma\|\rho)$ 
are strictly increasing on $(0,1)$.
\end{enumerate}
\end{thm}
\begin{proof}
The second inequality in \eqref{testdiv standard bounds} was already stated in 
Lemma \ref{lemma:Renyi order}. 
Note that by definition,  
\begin{align*}
\frac{\alpha}{1-\alpha}H_r(\rho\|\sigma)\ge -r+D_{\alpha}(\rho\|\sigma),
\end{align*}
for every $\alpha\in(0,1)$ and every $r\in\bR$.
Choosing $r=D_{\alpha}(\rho\|\sigma)/2$ yields, 
by \eqref{testdiv formula}, 
the first inequality in \eqref{testdiv standard bounds}.

The implications \ref{strict ineq1}$\imp$\ref{strict ineq2}$\imp$\ref{strict ineq3} 
are trivial. 
We prove the implication \ref{strict ineq3}$\imp$\ref{strict ineq4} by contraposition.
Note that by Lemma \ref{lemma:alpha swap2},
$\frac{1}{n}D_{\alpha}^{\test}(\rho^{\otimes n}\|\sigma^{\otimes n})
<D_{\alpha}(\rho\|\sigma)$ for some $\alpha\in(0,1)$ and some $n\in\bN$
if and only if 
$\frac{1}{n}D_{1-\alpha}^{\test}(\sigma^{\otimes n}\|\rho^{\otimes n})
<D_{1-\alpha}(\sigma\|\rho)$ for the same $\alpha\in(0,1)$ and $n\in\bN$.
Hence, the roles
of $\rho$ and $\sigma$ are symmetric, and therefore we may assume without loss of generality that
it is $\beta\mapsto D_{\beta}(\rho\|\sigma)$ that is not strictly increasing on $(0,1)$. Then we can write
$\rho=\sum_{i=1}^m\kappa s_i P_i$ as in 
Lemma \ref{lemma:Dalpha monotone}. It follows immediately that 
\begin{align*}
0=\log\Tr\rho=\log\kappa+\log\sum_{i=1}^m s_i\Tr P_i
=\log\kappa+\log\Tr\rho^0\sigma,
\end{align*}
and for every $n\in\bN$ and $\alpha\in(0,1)$, the test $T_n:=(\rho^0)^{\otimes n}$ gives
\begin{align*}
D_{\alpha}(\T_n(\rho^{\otimes n})\|\T_n(\sigma^{\otimes n}))
&=
\frac{1}{\alpha-1}\log\bz
\bz\Tr(\rho^0)^{\otimes n}\rho^{\otimes n}\jz^{\alpha}
\bz\Tr(\rho^0)^{\otimes n}\sigma^{\otimes n}\jz^{1-\alpha}\jz\\
&=
-n\log\Tr\rho^0\sigma=n\log\kappa=nD_{\alpha}(\rho|\sigma).
\end{align*}
Thus,
$\frac{1}{n}D_{\alpha}^{\test}(\rho^{\otimes n}\|\sigma^{\otimes n})
=D_{\alpha}(\rho\|\sigma)$ for every $\alpha\in(0,1)$ and every $n\in\bN$,
as required.

Finally, we assume \ref{strict ineq4} and prove \ref{strict ineq1}.
Note that the assumption that $\beta\mapsto D_\beta(\rho\|\sigma)$ is strictly increasing on $(0,1)$ is in fact equivalent to the strict inequality 
$D_0(\rho\|\sigma)<D_{\beta}(\rho\|\sigma)$ for all (equivalently, some) $\beta\in(0,1)$;
see Lemma \ref{lemma:Dalpha monotone}. Note that it also implies that $\rho\ne\sigma$,
whence
\begin{align*}
D_{\alpha}(\rho\|\sigma)>0,\ds\ds\ds\alpha\in(0,1);
\end{align*}
see Remark \ref{rem:strict positivity}. Let us fix an $\alpha\in(0,1)$.
By Lemma \ref{lemma:phi properties}, 
$r\mapsto H_r(\rho\|\sigma)-\frac{1-\alpha}{\alpha}r$ is a 
finite-valued convex, and hence continuous, function on 
the non-trivial interval $[D_0(\rho\|\sigma),D_{\alpha}(\rho\|\sigma)]$. Hence, 
by \eqref{testdiv formula} in Theorem \ref{thm:testdiv}, 
to prove the strict inequality $\oll{D}^{\test}_{\alpha}(\rho\|\sigma)<
D_{\alpha}(\rho\|\sigma)$,
we only need to show that 
\begin{align}\label{strict ineq proof2}
H_{r}(\rho\|\sigma)
-
\frac{1-\alpha}{\alpha}r<0
\end{align}
for $r=D_{\alpha}(\rho\|\sigma)$.
By \eqref{Hoeffding def}, this is equivalent to 
\begin{align*}
\sup_{\beta\in(0,1)}\frac{\beta-1}{\beta}\left[D_{\alpha}(\rho\|\sigma)-D_{\beta}(\rho\|\sigma)\right]<\frac{1-\alpha}{\alpha}D_{\alpha}(\rho\|\sigma).
\end{align*}
Since $\beta\mapsto D_{\beta}(\rho\|\sigma)$ is monotone increasing, 
it is sufficient to prove that 
\begin{align*}
\sup_{\beta\in[\alpha,1)}\frac{\beta-1}{\beta}
\left[D_{\alpha}(\rho\|\sigma)-D_{\beta}(\rho\|\sigma)\right]
<
\frac{1-\alpha}{\alpha}D_{\alpha}(\rho\|\sigma),
\end{align*}
or equivalently, that
\begin{align}\label{strict inequality proof1}
\inf_{\beta\in[\alpha,1)}
\Bigg\{
\underbrace{\frac{1-\alpha}{\alpha}D_{\alpha}(\rho\|\sigma)-
\frac{\beta-1}{\beta}
\left[D_{\alpha}(\rho\|\sigma)-D_{\beta}(\rho\|\sigma)\right]}_{=:g(\beta)}\Bigg\}>0.
\end{align}
For any $\beta\in[\alpha,1)$, we have 
\begin{align*}
g(\beta)&=
\underbrace{\frac{1-\alpha}{\alpha}D_{\alpha}(\rho\|\sigma)}_{=D_{1-\alpha}(\sigma\|\rho)}
+\frac{1-\beta}{\beta}D_{\alpha}(\rho\|\sigma)
-\underbrace{\frac{1-\beta}{\beta}D_{\beta}(\rho\|\sigma)}_{=D_{1-\beta}(\sigma\|\rho)}\\
&=
\underbrace{\frac{1-\beta}{\beta}D_{\alpha}(\rho\|\sigma)}_{>0}+
\underbrace{D_{1-\alpha}(\sigma\|\rho)-D_{1-\beta}(\sigma\|\rho)}_{\ge 0}\\
&>0,
\end{align*}
where the equality in the second line is by \eqref{alpha swap}.
Moreover, 
\begin{align*}
g(1):=\lim_{\beta\nearrow 1}g(\beta)=D_{1-\alpha}(\sigma\|\rho)-D_0(\sigma\|\rho)>0,
\end{align*}
where the strict inequality follows by the assumption on
$\beta\mapsto D_\beta(\sigma\|\rho)$.
Hence, $g$ is a strictly positive continuous function on the compact interval $[\alpha,1]$, 
and therefore \eqref{strict inequality proof1} holds.
\end{proof}

A trivial reformulation of the above gives the following:

\begin{cor}\label{cor:oll test equal Petz}
Let $\rho,\sigma\in\S(\hil)$ with spectral decompositions
$\rho=\sum_{i=1}^rr_iP_i$, 
$\sigma=\sum_{i=1}^ms_iQ_i$.
The following are equivalent:
\begin{enumerate}
\item\label{strict ineq1-1}
$\oll{D}^{\test}_{\alpha}(\rho\|\sigma)=
D_{\alpha}(\rho\|\sigma)$ for some $\alpha\in(0,1)$.
\item\label{strict ineq2-1}
$\oll{D}^{\test}_{\alpha}(\rho\|\sigma)=
D_{\alpha}(\rho\|\sigma)$ for every $\alpha\in(0,1)$.
\item\label{strict ineq3-1}
$\frac{1}{n}D_{\alpha}^{\test}(\rho^{\otimes n}\|\sigma^{\otimes n})
=D_{\alpha}(\rho\|\sigma)$ for every $\alpha\in(0,1)$ and every $n\in\bN$.
\item\label{strict ineq4-1}
One of the following two conditions is satisfied:
a) there exist a number $\kappa>0$ and projections $P_i'\le Q_i$, $i\in[m]$, such that 
$\rho=\sum_{i=1}^m\kappa s_iP_i'$, or 
b) there exist a number $\eta>0$ and projections $Q_i'\le P_i$, $i\in[r]$, such that 
$\sigma=\sum_{i=1}^r\eta r_iQ_i'$.
\end{enumerate}
Moreover, if any (and hence all) of the above holds then 
\begin{align}\label{all equal}
\oll{D}^{\test}_{\alpha}(\rho\|\sigma)=\hat{D}^{\test}_{\alpha}(\rho\|\sigma)=
D_{\alpha}(\rho\|\sigma),\ds\ds\ds\alpha\in(0,1).
\end{align}
\end{cor}
\begin{proof}
Since the above points are just the negations of the corresponding points in 
Theorem \ref{thm:testdiv standard bounds},
their equivalence is immediate from Theorem \ref{thm:testdiv standard bounds}
(in the last point we also use Lemma \ref{lemma:Dalpha monotone}). 
The equalities in \eqref{all equal} are immediate from the above and the inequalities
in Lemma \ref{lemma:Renyi order}.
\end{proof}

\begin{rem}
Note that the first inequality in \eqref{testdiv standard bounds}
can also be obtained from \eqref{alt repr4} by putting 
$t=\alpha$ instead of the optimization over $t$.
\end{rem}

\begin{rem}
Combining the first inequality in \eqref{testdiv standard bounds} with the inequalities in 
Lemma \ref{lemma:Renyi order} yields
\begin{align*}
\half D_{\alpha}(\rho\|\sigma)\le 
\oll{D}_{\alpha}^{\meas}(\rho\|\sigma),
\ds\ds\ds \alpha\in(0,1).
\end{align*}
Taking into account \eqref{Renyi order3} and \eqref{alpha swap}, the above is equivalent to 
\begin{align}\label{standard sandwiched ineq}
\half D_{\alpha}(\rho\|\sigma)\le D_{\alpha}\nw(\rho\|\sigma),
\ds\ds\ds \alpha\in[1/2,1).
\end{align}
This is a non-trivial inequality between the standard and the sandwiched R\'enyi divergences. 
However, it was shown in \cite[Corollary 2.3]{ItenRenesSutter} that
$\alpha D_{\alpha}(\rho\|\sigma)\le D_{\alpha}\nw(\rho\|\sigma)$
for $\alpha\in[0,1]$, which gives a stronger inequality 
than \eqref{standard sandwiched ineq} when $\alpha\in(1/2,1)$.
\end{rem}

\begin{rem}
Assume that $D_\alpha(\rho\|\sigma)=D_0(\rho\|\sigma)>0$, $\alpha\in(0,1)$; 
in particular, $\beta\mapsto D_{\beta}(\rho\|\sigma)$ is constant, while 
$\beta\mapsto D_\beta(\sigma\|\rho)={\beta\over1-\beta}D_0(\rho\|\sigma)$
 is strictly increasing on $(0,1)$ (see \eqref{alpha swap} for the equality). Then
\begin{align*}
H_r(\rho\|\sigma)
=
\begin{cases}+\infty, & r<D_0(\rho\|\sigma),\\ 0, & r\ge D_0(\rho\|\sigma), 
\end{cases}
\end{align*}
and thus, by Theorem \ref{thm:testdiv}, 
$\overline D_\alpha^\test(\rho\|\sigma)=D_0(\rho\|\sigma)=D_\alpha(\rho\|\sigma)$, $\alpha\in(0,1)$.
On the other hand, 
$H_{D_\alpha(\rho\|\sigma)}(\rho\|\sigma)=H_{D_0(\rho\|\sigma)}(\rho\|\sigma)=0
<{1-\alpha\over\alpha}D_\alpha(\rho\|\sigma)$ for all $\alpha\in(0,1)$.
Hence, \eqref{strict ineq proof2} and the strict increasing of
$\beta\mapsto D_\beta(\sigma\|\rho)$ are not sufficient to conclude 
$\oll{D}_\alpha^{\test}(\sigma\|\rho)<D_\alpha(\sigma\|\rho)$, $\alpha\in(0,1)$, 
in the proof of \ref{strict ineq4}$\imp$\ref{strict ineq1} in Theorem \ref{thm:testdiv standard bounds}.
\end{rem}

\begin{rem}
Related to the above remark, assume that $0<D_0(\rho\|\sigma)$. Then, by 
\eqref{Hr at r=D_0} and Theorem \ref{thm:testdiv}, we get that 
for a given $\alpha\in(0,1)$,
\begin{align}
\oll{D}_{\alpha}^{\test}(\rho\|\sigma)=D_0(\rho\|\sigma)
&\ds\iff\ds
\frac{1-\alpha}{\alpha}D_0(\rho\|\sigma)\ge H_{D_0(\rho\|\sigma)}(\rho\|\sigma)
=
-\psi(0)-\psi'(0)\nn\\
&\ds\iff\ds
\alpha\le \frac{\psi(0)}{\psi'(0)+2\psi(0)}.\label{alpha upper bound}
\end{align}
In particular, if $\beta\mapsto D_{\beta}(\rho\|\sigma)$ is strictly increasing, then 
$\oll{D}_{\alpha}^{\test}(\rho\|\sigma)=D_0(\rho\|\sigma)<
D_{\alpha}(\rho\|\sigma)$ for every $\alpha$ as in \eqref{alpha upper bound}.
\end{rem}
\medskip

We close this section with a few further corollaries of Theorems \ref{thm:testdiv}
and \ref{thm:testdiv standard bounds}.

\begin{cor}\label{cor:regularized strictly positive}
$\oll{D}_{\alpha}^{\test}$ is strictly positive for all $\alpha\in(0,1)$ in the sense that 
$\oll{D}_{\alpha}^{\test}(\rho\|\sigma)\ge 0$ for all $\rho,\sigma\in\S(\hil)$, with equality if and only if $\rho=\sigma$.
\end{cor}
\begin{proof}
Immediate from Theorem \ref{thm:testdiv standard bounds} and the strict positivity of 
$D_{\alpha}$ (see Remark \ref{rem:strict positivity}).
\end{proof}

\begin{cor}\label{cor:wa}
$\oll{D}_{\alpha}^{\test}$ is \ki{weakly additive} in the sense that 
\begin{align*}
\oll{D}_{\alpha}^{\test}\bz\rho^{\otimes k}\|\sigma^{\otimes k}\jz
=
k\oll{D}_{\alpha}^{\test}(\rho\|\sigma),\ds\ds\ds k\in\bN.
\end{align*}
\end{cor}
\begin{proof}
Immediate from \eqref{testdiv limit}.
\end{proof}

\begin{cor}
For any states $\rho,\sigma\in\S(\hil)$ and any $\alpha,\beta\in(0,1)$, 
\begin{align}\label{Renyi ineq}
D_{\alpha}(\rho\|\sigma)\ge\frac{\alpha(1-\beta)}{\alpha-2\alpha \beta+\beta}D_{\beta}(\rho\|\sigma).
\end{align}
\end{cor}
\begin{proof}
Since $D_{\alpha}(\rho\|\sigma)\ge \oll{D}_{\alpha}^{\test}(\rho\|\sigma)$,
Theorem \ref{thm:testdiv} gives 
\begin{align}\label{Renyi ineq proof}
\frac{1-\alpha}{\alpha}D_{\alpha}(\rho\|\sigma)\ge H_{D_{\alpha}(\rho\|\sigma)}(\rho\|\sigma)
\ge
\frac{\beta-1}{\beta}\left[D_{\alpha}(\rho\|\sigma)-D_{\beta}(\rho\|\sigma)\right],
\end{align}
where the second inequality is obvious by definition. 
Multipliying by $\alpha \beta$, and noting that $(\alpha+\beta)/2\ge\sqrt{\alpha \beta}
> \alpha \beta$, so that division by $\alpha-2\alpha \beta+\beta$ does not change the direction of
the inequality, we get 
\eqref{Renyi ineq} from \eqref{Renyi ineq proof} by a simple rearrangement.
\end{proof}
\medskip

The \ki{Chernoff divergence} of $\rho$ and $\sigma$ is defined as
\begin{align*}
C(\rho\|\sigma):=-\min_{0\le\alpha\le 1}\log\Tr\rho^{\alpha}\sigma^{1-\alpha}
=-\min_{0\le\alpha\le 1}\psi(\alpha).
\end{align*}
It is not too difficult to see from its operational interpretation in hypothesis testing 
\cite{Aud,NSz}
that 
\begin{align}\label{Chernoff char}
H_r(\rho\|\sigma)=r\ds\iff\ds r=C(\rho\|\sigma).
\end{align}

\begin{cor}\label{C-III.29}
For any $\rho,\sigma\in\S(\hil)$,
\begin{align*}
\oll{D}^{\test}_{1/2}(\rho\|\sigma)=C(\rho\|\sigma).
\end{align*}
\end{cor}
\begin{proof}
Immediate from Theorem \ref{thm:testdiv} and \eqref{Chernoff char}. 
Alternatively, it follows from the operational interpretation of the Chernoff divergence
\cite{Aud,NSz} and Proposition \ref{lemma:alt repr},
or directly from \eqref{alt repr4}.
\end{proof}
\medskip

Recall from \eqref{testdiv equalities} that the regularized measured and 
the regularized test-measured R\'enyi $\alpha$-divergences all coincide for 
$\alpha\ge 1$ and, moreover, for $\alpha=1$ they all yield the Umegaki relative entropy. 
In the following proposition we will use the \ki{max-relative entropy} \cite{Datta,RennerPhD}
of two states $\rho$ and $\sigma$, defined as 
\begin{align*}
D_{\max}(\rho\|\sigma):=
\begin{cases}
\log\min\{\lambda\ge 0:\,\rho\le \lambda\sigma\},&\rho^0\le\sigma^0,\\
+\infty,&\text{otherwise}.
\end{cases}
\end{align*}

\begin{prop}\label{prop:limits}
For $\rho,\sigma\in\S(\hil)$,
let $\divv_{\alpha}(\rho\|\sigma)$ denote any of the (regularized) measured or test-measured 
R\'enyi $\alpha$-divergences, as in Lemma \ref{lemma:alpha swap2}.
Then $\divv_{\alpha}(\rho\|\sigma)$ is increasing in $\alpha\in(0,+\infty)$, and
\begin{enumerate}
\item\label{0 limit}
$\lim_{\alpha\searrow 0}\divv_{\alpha}(\rho\|\sigma)=D_0(\rho\|\sigma)$;
\item\label{1 limit 1}
$\lim_{\alpha\to1}D_\alpha^\test(\rho\|\sigma)=D^\test(\rho\|\sigma)$;
\item\label{1 limit 2}
 $\lim_{\alpha\to1}D_\alpha^\meas(\rho\|\sigma)=D^\meas(\rho\|\sigma)$;
\item\label{1 limit 3} 
$\lim_{\alpha\to1}\oll{D}_\alpha^\test(\rho\|\sigma)
=\lim_{\alpha\to1}\hat D_\alpha^\test(\rho\|\sigma)
=\lim_{\alpha\to1}\oll{D}_{\alpha}^\meas(\rho\|\sigma)=D(\rho\|\sigma)$;
\item\label{infty limit} 
$\lim_{\alpha\to +\infty}\divv_{\alpha}(\rho\|\sigma)=
D_{\max}(\rho\|\sigma)$.
\end{enumerate}
\end{prop}

\begin{proof}
Monotonicity in $\alpha$ is obvious from the monotonicity of the classical R\'enyi divergences in 
$\alpha$; see, e.g., Lemma \ref{lemma:Dalpha monotone}.

\ref{0 limit}:\enspace
With the test $T:=\rho^0$ we have 
\begin{align*}
D_0(\rho\|\sigma)
=
D_0(\T(\rho)\|\T(\sigma))
=
\lim_{\alpha\searrow 0}D_{\alpha}(\T(\rho)\|\T(\sigma))
\le
\lim_{\alpha\searrow 0}D_{\alpha}^{\test}(\rho\|\sigma),
\end{align*}
where the first equality is straightforward to verify, and the second equality and the inequality are by definition. The assertion then follows from this and the inequalities in 
Lemma \ref{lemma:Renyi order}, except when 
$\divv_{\alpha}(\rho\|\sigma)=\oll{D}_{\alpha}^{\test}(\rho\|\sigma)$, in which case it follows immediately from \eqref{testdiv standard bounds}.

\ref{1 limit 1} and \ref{1 limit 2}:\enspace
We have
\begin{align*}
\sup_{\alpha\in(0,1)}D_\alpha^\test(\rho\|\sigma)
&=\sup_{\alpha\in(0,1)}\sup_{T\in\bT(\hil)}D_\alpha(\T(\rho)\|\T(\sigma))\\
&=\sup_{T\in\bT(\hil)}\sup_{\alpha\in(0,1)}D_\alpha(\T(\rho)\|\T(\sigma)) \\
&=\sup_{T\in\bT(\hil)}D(\T(\rho)\|\T(\sigma))\\
&=D^\test(\rho\|\sigma),
\end{align*}
which gives 
$\lim_{\alpha\nearrow1}D_\alpha^\test(\rho\|\sigma)=D^\test(\rho\|\sigma)$. 
An analogous argument gives 
$\lim_{\alpha\nearrow1}D_\alpha^\meas(\rho\|\sigma)=D^\meas(\rho\|\sigma)$. 

If $\rho^0\nleq\sigma^0$ then 
$\lim_{\alpha\searrow1}D_\alpha^\test(\rho\|\sigma)=+\infty=D^\test(\rho\|\sigma)$
and
$\lim_{\alpha\searrow1}D_\alpha^\meas(\rho\|\sigma)=+\infty=D^\meas(\rho\|\sigma)$ 
hold trivially, and hence for the rest we assume 
$\rho^0\le\sigma^0$. Then 
\begin{align*}
\lim_{\alpha\searrow 1}D_\alpha^\test(\rho\|\sigma)
&=
\inf_{\alpha>1}D_\alpha^\test(\rho\|\sigma)\\
&=\inf_{\alpha>1}\sup_{T\in\bT(\hil)}D_\alpha(\T(\rho)\|\T(\sigma))\\
&=\sup_{T\in\bT(\hil)}\inf_{\alpha>1}D_\alpha(\T(\rho)\|\T(\sigma)) \\
&=\sup_{T\in\bT(\hil)}D(\T(\rho)\|\T(\sigma))\\
&=D^\test(\rho\|\sigma),
\end{align*}
where the third equality can be seen from the minimax theorem in
\cite[Corollary A.2]{MH}, since $T\mapsto D_\alpha(\T(\rho)\|\T(\sigma))$ is continuous
on the compact space $\bT(\hil)$. 
An analogous argument gives that 
$\lim_{\alpha\searrow 1}D_\alpha^\meas(\rho\|\sigma)=
D^\meas(\rho\|\sigma)$, using 
\eqref{meas def} and 
the continuity of 
$M\mapsto D_\alpha(\M(\rho)\|\M(\sigma))$ on the compact space
$\mathrm{PVM}_1(\hil)$.

\ref{1 limit 3}:\enspace
By \eqref{testdiv equalities} and Lemmas \ref{lemma:Dalpha monotone} and \ref{lemma:Renyi order},
it suffices to prove that
$\lim_{\alpha\nearrow 1}\oll{D}_\alpha^\test(\rho\|\sigma)\ge D(\rho\|\sigma)$. 
Since this holds trivially when $\rho=\sigma$, for the rest we assume that $\rho\ne\sigma$, or equivalently, that $D(\rho\|\sigma)>0$. 
Note that $H_r(\rho\|\sigma)/r>0\iff r\in(0,D(\rho\|\sigma))$, 
according to \eqref{Hr null},
and hence for a given $r\in(0,+\infty)$, there exists an 
$\alpha\in(0,1)$ such that $H_r(\rho\|\sigma)/r\ge (1-\alpha)/\alpha$
if and only if $r\in(0,D(\rho\|\sigma))$.
Using then Theorem \ref{thm:testdiv} we get 
\begin{align*}
\lim_{\alpha\nearrow 1}\oll{D}_\alpha^\test(\rho\|\sigma)
=
\sup_{\alpha\in(0,1)}\oll{D}_{\alpha}^{\test}(\rho\|\sigma)
=
\sup_{\alpha\in(0,1)}\sup\left\{r:\,\frac{H_r(\rho\|\sigma)}{r}\ge\frac{1-\alpha}{\alpha}\right\}
=
D(\rho\|\sigma),
\end{align*}
where the first equality is trivial from the previously stated monotonicity in $\alpha$.

\ref{infty limit}:\enspace
The assertion is trivial when $\rho^0\nleq\sigma^0$, so for the rest we assume that 
$\rho^0\le\sigma^0$. 
If $\lambda>0$ is such that $\rho\le \lambda\sigma$ then 
$\divv_{\alpha}(\rho\|\sigma)\le \lambda$ is easy to see 
for every $\alpha>1$, which implies
\begin{align}\label{infty proof1}
\divv_{\alpha}(\rho\|\sigma)\le D_{\max}(\rho\|\sigma),\ds\ds\ds\alpha>1.
\end{align}
It is easy to verify that 
\begin{align}
D_{\max}(\rho\|\sigma)
&=
\log\max\left\{\frac{\Tr \rho T}{\Tr\sigma T}:\,T\in\bT(\hil),
\,\Tr\sigma T>0\right\};
\label{LP duality}
\end{align}
see, e.g., \cite[Corollary III.45]{Mosonyi_sc_2021}.
Choosing a $T$ that attains the maximum in \eqref{LP duality} gives that 
\begin{align}
\lim_{\alpha\to +\infty}D_{\alpha}^{\test}(\rho\|\sigma)
\ge
\lim_{\alpha\to +\infty}D_{\alpha}(\T(\rho)\|\T(\sigma))
=
\log\frac{\Tr\rho T}{\Tr\sigma T}=D_{\max}(\rho\|\sigma).
\label{infty proof2}
\end{align} 
Choosing for every $n\in\bN$ the test $T_n:=T^{\otimes n}$ with the above $T$ yields
for every $\alpha>1$,
\begin{align*}
\frac{1}{n}D_{\alpha}^{\test}(\rho^{\otimes n}\|\sigma^{\otimes n})
&\ge
\frac{1}{n}\frac{1}{\alpha-1}\log\bz\bz\Tr\rho^{\otimes n}T_n\jz^{\alpha}
\bz\Tr\sigma^{\otimes n}T_n\jz^{1-\alpha}\jz\\
&=
\frac{\alpha}{\alpha-1}\log\frac{\Tr\rho T}{\Tr\sigma T}
+\frac{1}{\alpha-1}\log\Tr\sigma T\\
&=
\frac{\alpha}{\alpha-1}D_{\max}(\rho\|\sigma)
+\frac{1}{\alpha-1}\log\Tr\sigma T.
\end{align*}
Taking first the limit $n\to+\infty$ and then the limit $\alpha\to +\infty$, we get 
\begin{align}\label{infty proof3}
\lim_{\alpha\to +\infty}\oll{D}_{\alpha}^{\test}(\rho\|\sigma)
\ge
D_{\max}(\rho\|\sigma).
\end{align}
Combining the lower and the upper bounds in \eqref{infty proof1}--\eqref{infty proof3} with the inequalities in Lemma \ref{lemma:Renyi order}
gives the desired statement. 
\end{proof}

\begin{rem}
An alternative proof for 
$\lim_{\alpha\nearrow 1}\oll{D}_\alpha^\test(\rho\|\sigma)=D(\rho\|\sigma)$, given in 
\ref{1 limit 3} of Proposition \ref{prop:limits}, can be obtained from \eqref{alt repr4}, or more precisely,
from \cite[Lemma 2]{Salzmann_Datta21}; our case 
corresponds to $\lim_{s\to 0}$ in the latter. In detail, 
\begin{align*}
\lim_{\alpha\nearrow 1}\oll{D}_\alpha^\test(\rho\|\sigma)
&=
\sup_{\alpha\in(0,1)}\oll{D}_\alpha^\test(\rho\|\sigma)\\
&=\sup_{\alpha\in(0,1)}\sup_{t\in(0,1)}
{(t-1)D_t(\rho\|\sigma)\over t(2-1/\alpha)-1}\\
&=\sup_{t\in(0,1)}\sup_{\alpha\in(0,1)}
{(t-1)D_t(\rho\|\sigma)\over t(2-1/\alpha)-1} \\
&=\sup_{t\in(0,1)}D_t(\rho\|\sigma)
=D(\rho\|\sigma),
\end{align*}
which from the second equality is the same proof as in 
\cite[Lemma 2]{Salzmann_Datta21}.
\end{rem}

\begin{rem}
It is known that 
\begin{align}\label{sandwiched infty limit}
D_{\infty}\nw(\rho\|\sigma):=\lim_{\alpha\to+\infty}D_{\alpha}\nw(\rho\|\sigma)
=
D_{\max}(\rho\|\sigma)
\end{align}
for any pair of states $\rho,\sigma$; see \cite[Theorem 5]{Renyi_new} or 
\cite[Sec.~4.2.4]{TomamichelBook}.
This of course also implies that 
$\lim_{\alpha\to+\infty}\oll{D}_{\alpha}^{\test}(\rho\|\sigma)=
\lim_{\alpha\to+\infty}\hat D_{\alpha}^{\test}(\rho\|\sigma)=
\lim_{\alpha\to+\infty}\oll{D}_{\alpha}^{\meas}(\rho\|\sigma)=
D_{\max}(\rho\|\sigma)$, if we use the identities in \eqref{testdiv equalities};
however, in the proof of Proposition \ref{prop:limits} \ref{infty limit} above, we did not rely on these.

Vice versa, we can get a new proof of \eqref{sandwiched infty limit}
from Proposition \ref{prop:limits} \ref{infty limit}
by only using the monotonicity of $D_{\alpha}\nw$ under CPTP maps for $\alpha>1$ as follows. 
First, it is straightforward to verify that 
$D_{\alpha}\nw(\rho\|\sigma)\le D_{\max}(\rho\|\sigma)$.
Second, using monotonicity gives 
$\lim_{\alpha\to+\infty}D_{\alpha}\nw(\rho\|\sigma)
\ge
\lim_{\alpha\to+\infty}D_{\alpha}^{\test}(\rho\|\sigma)
=
D_{\max}(\rho\|\sigma)$,
where the equality is due to 
Proposition \ref{prop:limits} \ref{infty limit}.
\end{rem}

\subsection{Analysis of $\hat D_{\alpha}^{\test}$}
\label{sec:an2}

Our goal in this section is to show that the strict inequality 
$\hat D_{\alpha}^{\test}(\rho\|\sigma)<\oll{D}_{\alpha}^{\meas}(\rho\|\sigma)$
can hold for $\alpha\in(0,1)$; in fact, we show that this is the case for
any $\alpha\in(0,1)$ whenever $\rho$ and $\sigma$ are commuting states
(and hence $\oll{D}_{\alpha}^{\meas}(\rho\|\sigma)=D_{\alpha}(\rho\|\sigma)$), 
that are not equal and have the same supports.
Interestingly, we will derive this from the strict inequality 
$\oll{D}_{\alpha}^{\test}(\rho\|\sigma)<D_{\alpha}(\rho\|\sigma)$,
(established in Theorem \ref{thm:testdiv standard bounds})
which might seem a bit counter-intuitive at first sight, since 
$\oll{D}_{\alpha}^{\test}(\rho\|\sigma)\le 
\hat D_{\alpha}^{\test}(\rho\|\sigma)$; \
moreover, this last inequality might be strict, as we show below.
The key to go from the strict upper bound on $\oll{D}_{\alpha}^{\test}$
to the strict upper bound on $\hat D_{\alpha}^{\test}$
is a simple observation, given in the following two lemmas:

\begin{lemma}\label{lemma:sup eq lim}
If $\frac{1}{n}D_{\alpha}^{\test}(\rho^{\otimes n}\|\sigma^{\otimes n})<\hat D_{\alpha}^{\test}(\rho\|\sigma)$ for all $n$, then 
$\oll{D}_{\alpha}^{\test}(\rho\|\sigma)=
\hat D_{\alpha}^{\test}(\rho\|\sigma)$;
or equivalently, if 
$\oll{D}_{\alpha}^{\test}(\rho\|\sigma)<\hat D_{\alpha}^{\test}(\rho\|\sigma)$
then there exists an $n\in\bN$ such that 
$\frac{1}{n}D_{\alpha}^{\test}(\rho^{\otimes n}\|\sigma^{\otimes n})=\hat D_{\alpha}^{\test}(\rho\|\sigma)$.
\end{lemma}
\begin{proof}
If $\frac{1}{n}D_{\alpha}^{\test}(\rho^{\otimes n}\|\sigma^{\otimes n})<\hat D_{\alpha}^{\test}(\rho\|\sigma)$ for all $n$, then
there exists a strictly increasing function 
$k:\,\bN\to\bN$ such that 
\begin{align*}
\hat D_{\alpha}^{\test}(\rho\|\sigma)
=
\lim_{n\to+\infty}\frac{1}{k(n)}D_{\alpha}^{\test}(\rho^{\otimes k(n)}\|\sigma^{\otimes k(n)})
\le 
\oll{D}_{\alpha}^{\test}(\rho\|\sigma),
\end{align*}
where the inequality is by definition. Since $\oll{D}_{\alpha}^{\test}(\rho\|\sigma)\le
\hat D_{\alpha}^{\test}(\rho\|\sigma)$ is obvious, the statement follows.
\end{proof}

\begin{lemma}\label{lemma:maximal testdiv smaller}
If $\rho,\sigma\in\S(\hil)$ and $\alpha\in(0,1)$ are such that 
\begin{enumerate}
\item\label{maximal testdiv smaller1}
$\oll{D}_{\alpha}^{\test}(\rho\|\sigma)<\oll{D}_{\alpha}^{\meas}(\rho\|\sigma)$;
\item\label{maximal testdiv smaller2}
$\frac{1}{n}D_{\alpha}^{\test}(\rho^{\otimes n}\|\sigma^{\otimes n})<
\oll{D}_{\alpha}^{\meas}(\rho\|\sigma)$, $n\in\bN$,
\end{enumerate}
then $\hat D_{\alpha}^{\test}(\rho\|\sigma)< \oll{D}_{\alpha}^{\meas}(\rho\|\sigma)$.
\end{lemma}
\begin{proof}
Assume that $\hat D_{\alpha}^{\test}(\rho\|\sigma)= \oll{D}_{\alpha}^{\meas}(\rho\|\sigma)$. Then, by \ref{maximal testdiv smaller2} above, 
$\frac{1}{n}D_{\alpha}^{\test}(\rho^{\otimes n}\|\sigma^{\otimes n})<
\oll{D}_{\alpha}^{\meas}(\rho\|\sigma)=
\hat D_{\alpha}^{\test}(\rho\|\sigma)$ for all $n$, i.e., 
the assumption in Lemma \ref{lemma:sup eq lim} holds, and therefore
$\oll{D}_{\alpha}^{\test}(\rho\|\sigma)
=\hat D_{\alpha}^{\test}(\rho\|\sigma)
=
\oll{D}_{\alpha}^{\meas}(\rho\|\sigma)$, 
contradicting \ref{maximal testdiv smaller1}.
\end{proof}

We have seen in Theorem \ref{thm:testdiv standard bounds} that \ref{maximal testdiv smaller1} in 
Lemma \ref{lemma:maximal testdiv smaller} holds for generic commuting pairs, and hence 
for the rest we focus on finite-copy bounds as in \ref{maximal testdiv smaller2}
of Lemma \ref{lemma:maximal testdiv smaller}.
First, we consider the case where
equality holds in \eqref{test-meas ineq}.

\begin{lemma}\label{lemma:test-meas eq}
Let $\rho,\sigma\in\S(\hil)$ and $\alpha\in(0,1)$.
If $D_{\alpha}^{\test}(\rho\|\sigma)=D_{\alpha}^{\meas}(\rho\|\sigma)$ holds, then
for any projection $P=T$ attaining the maximum in \eqref{test-div pr2}, we have 
\begin{align}\label{test-meas eq}
(\Tr\sigma P)P\rho P=(\Tr\rho P)P\sigma P,\ds\ds\ds
(\Tr\sigma P^{\perp})P^{\perp}\rho P^{\perp}=(\Tr\rho P^{\perp})P^{\perp}\sigma P^{\perp}.
\end{align}
If, moreover, $\rho\ne \sigma$ then $\rho^0\le\sigma^0$ implies $P\sigma P\ne 0$, $P^{\perp}\sigma P^{\perp}\ne 0$,
and $\rho^0\ge\sigma^0$ implies $P\rho P\ne 0$, $P^{\perp}\rho P^{\perp}\ne 0$.
\end{lemma}
\begin{proof}
Assume that $D_{\alpha}^{\test}(\rho\|\sigma)=D_{\alpha}^{\meas}(\rho\|\sigma)$
holds, and let $P=T$ be a projection attaining the maximum in \eqref{test-div pr2}.
Let $(e_i)_{i=1}^r$ be an orthonormal basis in $\ran P$, and $(e_i)_{i=r+1}^d$ be an 
orthonormal basis in $\ran P^{\perp}$. By H\"older's inequality,
\begin{align}
\sum_{i=1}^r\inner{e_i}{\rho e_i}^{\alpha}\inner{e_i}{\sigma e_i}^{1-\alpha}
&\le \bz\sum_{i=1}^r\inner{e_i}{\rho e_i}\jz^{\alpha}
\bz\sum_{i=1}^r\inner{e_i}{\sigma e_i}\jz^{1-\alpha}
=
(\Tr\rho P)^{\alpha}(\Tr\sigma P)^{1-\alpha},\label{test-meas eq proof}\\
\sum_{i=r+1}^d\inner{e_i}{\rho e_i}^{\alpha}\inner{e_i}{\sigma e_i}^{1-\alpha}
&\le \bz\sum_{i=r+1}^d\inner{e_i}{\rho e_i}\jz^{\alpha}
\bz\sum_{i=r+1}^d\inner{e_i}{\sigma e_i}\jz^{1-\alpha}
=
(\Tr\rho P^{\perp})^{\alpha}(\Tr\sigma P^{\perp})^{1-\alpha}.\label{test-meas eq proof1}
\end{align}
Thus,
\begin{align*}
D_\alpha^\meas(\rho\|\sigma)
&\ge{1\over\alpha-1}\log\sum_{i=1}^d\inner{e_i}{\rho e_i}^\alpha
\inner{e_i}{\sigma e_i}^{1-\alpha} \\
&\ge{1\over\alpha-1}\log\bz(\Tr\rho P)^\alpha(\Tr\sigma P)^{1-\alpha}
+(\Tr\rho P^\perp)^\alpha(\Tr\sigma P^\perp)^{1-\alpha}\jz \\
&=D_\alpha^{\test}(\rho\|\sigma),
\end{align*}
where the first inequality is by definition, and the second one follows from 
\eqref{test-meas eq proof}--\eqref{test-meas eq proof1}.
Therefore, the assumed equality implies that the inequalities in 
\eqref{test-meas eq proof}--\eqref{test-meas eq proof1} hold as equalities.
Using the characterization of equality in H\"older's inequality, 
we get 
\begin{align}\label{test-meas eq proof2}
(\Tr\sigma P)\inner{e_i}{\rho e_i}=(\Tr\rho P)\inner{e_i}{\sigma e_i},\ds\ds\ds
i\in[r],
\end{align}
and similarly for $P^{\perp}$ in place of $P$ and $i=r+1,\ldots,d$.
Since \eqref{test-meas eq proof2} holds for every orthonormal basis in $\ran P$, and any unit vector $\psi\in\ran P$ can be extended to such an orthonormal basis, we get that 
\begin{align*}
(\Tr\sigma P)\inner{\psi}{\rho \psi}=(\Tr\rho P)\inner{\psi}{\sigma \psi},\ds\ds\ds
\psi\in\ran P,
\end{align*}
or equivalently,
\begin{align*}
\inner{\psi}{\left[(\Tr\sigma P)P\rho P-(\Tr\rho P)P\sigma P\right]\psi}=0,\ds\ds\ds
\psi\in\hil,
\end{align*}
and similarly with $P^{\perp}$ in place of $P$, which is exactly
\eqref{test-meas eq}.

Consider now the case
$\rho\ne\sigma$ and $\rho^0\le\sigma ^0$.
Assume that $P\sigma P=0$, or equivalently, 
$\sigma^0\le P^{\perp}$. Then we also have $\rho^0\le P^{\perp}$, and
\begin{align*}
D_{\alpha}^{\test}(\rho\|\sigma)
&=
D_{\alpha}(\P(\rho)\|\P(\sigma))
=
\frac{1}{\alpha-1}\log\underbrace{(\Tr\rho P^{\perp})^{\alpha}}_{=1}\underbrace{(\Tr\sigma 
P^{\perp})^{1-\alpha}}_{=1}
=0<D_{\alpha}^{\test}(\rho\|\sigma),
\end{align*}
where $\P$ is defined analogously to $\T$ in \eqref{postmeas}, and
the last inequality is due to the strict positivity of $D_{\alpha}^{\test}$ given in 
Lemma \ref{lemma:testdiv strictly positive}. This is a contradiction, and hence 
$P\sigma P=0$ cannot hold. Replacing $P$ with $P^{\perp}$ in the same argument gives that 
$P^{\perp}\sigma P^{\perp}=0$ cannot hold, either.
The assertion about the case $\rho^0\ge\sigma^0$ follows in the same way.
\end{proof}

\begin{cor}\label{cor:classical equality}
Assume that $\rho,\sigma\in\S(\hil)$ commute, and hence can be written as 
in \eqref{commuting pair}.
Then the following are equivalent:
\begin{enumerate}
\item\label{classical equality 0}
$D_{\alpha}^{\test}(\rho\|\sigma)=D_{\alpha}(\rho\|\sigma)$ for all $\alpha\in(0,1)$;
\item\label{classical equality 1}
$D_{\alpha}^{\test}(\rho\|\sigma)=D_{\alpha}(\rho\|\sigma)$ for some $\alpha\in(0,1)$;
\item\label{classical equality 2}
there exists a subset $\Omega_0\subseteq\Omega$ 
such that 
\begin{align}\label{classical equality}
\sigma(\Omega_0)\rho(\omega)=\rho(\Omega_0)\sigma(\omega),\ds\omega\in\Omega_0,\ds\ds\ds
\sigma(\Omega_1)\rho(\omega)=\rho(\Omega_1)\sigma(\omega),\ds\omega\in\Omega_1,
\end{align}
where $\Omega_1:=\Omega\setminus\Omega_0$, and 
$\sigma(\Omega_0):=\sum_{\omega\in\Omega_0}\sigma(\omega)$, etc.
\end{enumerate}
Moreover, if any (and hence all) of the above holds then 
\begin{align*}
D_{\alpha}^{\test}(\rho\|\sigma)=
\hat D_{\alpha}^{\test}(\rho\|\sigma)=
D_{\alpha}(\rho\|\sigma),\ds\ds\ds\alpha\in(0,1).
\end{align*}
\end{cor}
\begin{proof}
The implication \ref{classical equality 0}$\imp$\ref{classical equality 1} is obvious.
By Remark \ref{rem:subalgebra}, an optimal projective test $P$ as in Lemma \ref{lemma:test-meas eq} can be written as $P=\sum_{\omega\in\Omega_0}\pr{\omega}$ for some 
$\Omega_0\subseteq\Omega$, and 
\ref{classical equality 1}$\imp$\ref{classical equality 2}
follows from Lemma \ref{lemma:test-meas eq}.

Assume now that \ref{classical equality 2} holds, and 
let $P:=\sum_{\omega\in\Omega_0}\pr{\omega}$. A straightforward computation yields that 
for any $\alpha\in(0,1)$,
\begin{align*}
D_{\alpha}(\P(\rho)\|\P(\sigma))
=
\frac{1}{\alpha-1}\log
\bz \rho(\Omega_0)^{\alpha}\sigma(\Omega_0)^{1-\alpha}
+
\rho(\Omega_1)^{\alpha}\sigma(\Omega_1)^{1-\alpha}
\jz
=
D_{\alpha}(\rho\|\sigma),
\end{align*}
whence $D_{\alpha}^{\test}(\rho\|\sigma)=D_{\alpha}(\rho\|\sigma)$.
Thus, \ref{classical equality 0} holds.

The last assertion is immediate from the above and the inequalities
$D_{\alpha}^{\test}(\rho\|\sigma)\le
\hat D_{\alpha}^{\test}(\rho\|\sigma)\le
D_{\alpha}(\rho\|\sigma)$, $\alpha\in(0,1)$.
\end{proof}

Corollaries \ref{cor:classical equality} and \ref{cor:oll test equal Petz}
yield immediately that $\oll{D}_{\alpha}^{\test}(\rho\|\sigma)$ and 
$\hat D_{\alpha}^{\test}(\rho\|\sigma)$ need not be equal in general:

\begin{thm}\label{thm:different reg}
Let $\rho=\sum_{\omega\in\Omega}\rho(\omega)\pr{\omega}$ and 
$\sigma=\sum_{\omega\in\Omega}\sigma(\omega)\pr{\omega}$ be unequal commuting states, and assume that \eqref{classical equality}
holds with $\rho(\Omega_k)\ne 0$, $\sigma(\Omega_k)\ne 0$, $k=0,1$. Then
\begin{align}\label{different test regularizations}
\oll{D}_{\alpha}^{\test}(\rho\|\sigma)<
\hat D_{\alpha}^{\test}(\rho\|\sigma)=
D_{\alpha}(\rho\|\sigma),\ds\ds\ds\alpha\in(0,1).
\end{align}
In particular, this is the case for unequal commuting qubit states with full support.
\end{thm}
\begin{proof}
The equality in \eqref{different test regularizations} is immediate from 
Corollary \ref{cor:classical equality}.
Note that $\frac{\rho(\Omega_0)}{\sigma(\Omega_0)}\ne\frac{\rho(\Omega_1)}{\sigma(\Omega_1)}$, 
since otherwise \eqref{classical equality} would imply $\rho=\sigma$, contrary to the
assumption. Thus, by \eqref{classical equality},
neither a) nor b) in \ref{strict ineq4-1} of 
Corollary \ref{cor:oll test equal Petz} hold, whence 
$\oll{D}_{\alpha}^{\test}(\rho\|\sigma)<D_{\alpha}(\rho\|\sigma)$,
$\alpha\in(0,1)$. The last assertion about qubit states is obvious.
\end{proof}
\medskip

Finally, we show (in Theorem \ref{thm:an2 main} below)
that \eqref{classical equality} is not only sufficient but also necessary 
for the equality 
$\hat D_{\alpha}^{\test}(\rho\|\sigma)=D_{\alpha}(\rho\|\sigma)$
in the case of generic commuting states, thereby giving a complete
and practically verifiable condition for this equality for such pairs of states. 
To this end, we fist prove the following two lemmas:

\begin{lemma}\label{lemma:single-copy equality}
For any $\rho,\sigma\in\S(\hil)$,
\begin{align}
D_{\alpha}^{\test}(\rho\|\sigma)\le
\begin{cases}
\oll{D}_{\alpha}^{\meas}(\rho\|\sigma),&\alpha\in(0,1/2)\cup(1/2,1),\\
D_{\alpha}(\rho\|\sigma),&\alpha=1/2.
\end{cases}
\label{classical sh strict ineq1}
\end{align}
Moreover, if $\rho^0\le\sigma^0$ 
and equality holds above for some $\alpha\in[1/2,1)$, or 
$\rho^0\ge \sigma^0$ 
and equality holds above for some $\alpha\in(0,1/2]$, 
then 
$\rho$ and $\sigma$ commute (whence they can be written as in \eqref{commuting pair}), and 
\ref{classical equality 2} of Corollary \ref{cor:classical equality} holds.
\end{lemma}
\begin{proof}
The inequality and the statement about the case of equality follows immediately 
from Proposition \ref{prop:single-copy equality meas} and the obvious inequality
$D_{\alpha}^{\test}(\rho\|\sigma)\le D_{\alpha}^{\meas}(\rho\|\sigma)$.
In particular, equality in \eqref{classical sh strict ineq1}
under the stated assumptions yields that 
$\rho$ and $\sigma$ commute, and 
\ref{classical equality 1} of Corollary \ref{cor:classical equality} holds, whence
\ref{classical equality 2} of the same corollary holds as well. 
\end{proof}

\begin{lemma}\label{lemma:classical strict inequality}
Let $\rho,\sigma\in\S(\hil)$ be unequal states with $\rho^0=\sigma^0$. For every 
$n\ge 2$,
\begin{align}
\frac{1}{n}D_{\alpha}^{\test}(\rho^{\otimes n}\|\sigma^{\otimes n})<
\begin{cases}
\oll{D}_{\alpha}^{\meas}(\rho\|\sigma),&\alpha\in(0,1/2)\cup(1/2,1),\\
D_{\alpha}(\rho\|\sigma),&\alpha=1/2.
\end{cases}
\label{classical strict ineq1}
\end{align}
\end{lemma}
\begin{proof}
Assume that equality holds in \eqref{classical strict ineq1} for some 
$n\in\bN$ and $\alpha\in(0,1)$. Then, by Lemma \ref{lemma:single-copy equality},
$\rho^{\otimes n}$ and $\sigma^{\otimes n}$ commute.
This of course implies that 
$\rho$ and $\sigma$ commute, and hence 
they can be written as in \eqref{commuting pair}.
Moreover, still by Lemma \ref{lemma:single-copy equality},
there exists a subset
$\Omega_{n,0}\subseteq\Omega^n$ 
such that
\begin{align*}
&\rho^{\otimes n}(\vecc{\omega})=c_{n,0}\sigma^{\otimes n}(\vecc{\omega}),\ds\ds
\vecc{\omega}\in\Omega_{n,0},
\\
&\rho^{\otimes n}(\vecc{\omega})=c_{n,1}\sigma^{\otimes n}(\vecc{\omega}),\ds\ds
\vecc{\omega}\in\Omega_{n,1},
\end{align*}
where $\Omega_{n,1}:=\Omega^n\setminus\Omega_{n,0}$, and
$c_{n,0}:=\rho^{\otimes n}\bz\Omega_{n,0}\jz/
\sigma^{\otimes n}\bz\Omega_{n,0}\jz$,
$c_{n,1}:=\rho^{\otimes n}\bz\Omega_{n,1}\jz/
\sigma^{\otimes n}\bz\Omega_{n,1}\jz$
are strictly positive constants.
This is equivalent to 
\begin{align}
\sum_{\omega\in\Omega}\tau_{n,\vecc{\omega}}(\omega)\log\frac{\rho(\omega)}{\sigma(\omega)}
&=\frac{1}{n}\log c_{n,0}=:\tilde c_{n,0},\ds\ds\ds\vecc{\omega}\in\Omega_{n,0},
\label{eq5}\\
\sum_{\omega\in\Omega}\tau_{n,\vecc{\omega}}(\omega)\log\frac{\rho(\omega)}{\sigma(\omega)}
&=\frac{1}{n}\log c_{n,1}=:\tilde c_{n,1},\ds\ds\ds\vecc{\omega}\in\Omega_{n,1},
\label{eq6}
\end{align}
where $\tau_{n,\vecc{\omega}}$ is the type of $\vecc{\omega}$, 
i.e., $\tau_{n,\vecc{\omega}}(\gamma):=\frac{1}{n}\#\{k\in[n]:\,\omega_k=\gamma\}$,
$\gamma\in\Omega$.
In particular, $\{\tau_{n,\vecc{\omega}}:\,\vecc{\omega}\in\Omega_{n,0}\}$ is in the intersection of the probability simplex with a hyperplane with normal vector 
$\bz\log\frac{\rho(\omega)}{\sigma(\omega)}\jz_{\omega\in\Omega}$, and 
$\{\tau_{n,\vecc{\omega}}:\,\vecc{\omega}\in\Omega_{n,1}\}$ 
is in the intersection of the probability simplex with a translate of that hyperplane.
Obviously, the two hyperplanes are different, i.e., $c_{n,0}\ne c_{n,1}$, since the opposite would yield that $\rho=\sigma$.

Since $\tau_{n,\vecc{\omega}}$ is the Dirac probability density function 
$\egy_{\{\omega\}}$ when $\vecc{\omega}=(\omega,\ldots,\omega)$ for some $\omega\in\Omega$,
we get
\begin{align*}
\log\frac{\rho(\omega)}{\sigma(\omega)}=\tilde c_{n,0},\ds\ds\ds(\omega,\ldots,\omega)\in\Omega_{n,0},
\ds\ds\ds\ds\ds
\log\frac{\rho(\omega)}{\sigma(\omega)}=\tilde c_{n,1},\ds\ds\ds(\omega,\ldots,\omega)\in\Omega_{n,1}.
\end{align*}
Combining this with \eqref{eq5}--\eqref{eq6}, we get that for any sequence 
$\vecc{\omega}\in\Omega^n$, 
\begin{align}
\sum_{\omega\in\Omega}\tau_{n,\vecc{\omega}}(\omega)\log\frac{\rho(\omega)}{\sigma(\omega)}
=
\frac{k}{n}\tilde c_{n,0}+\frac{n-k}{n}\tilde c_{n,1},
\label{eq7}
\end{align}
where $k:=\#\{j\in[n]:\,(\omega_j,\ldots,\omega_j)\in\Omega_{n,0}\}$. 
If $n>1$ then there exist sequences $\vecc{\omega}$ containing elements from 
both $\Omega_{n,0}$ and $\Omega_{n,1}$, whence 
the RHS of \eqref{eq7} is strictly between 
$\tilde c_{n,0}$ and $\tilde c_{n,1}$. By \eqref{eq5}--\eqref{eq6}, this means that 
$\vecc{\omega}$ is neither in $\Omega_{n,0}$ nor in $\Omega_{n,1}=\Omega^n\setminus\Omega_{n,0}$, a contradiction.
\end{proof}

\begin{rem}
Note that the support condition in Lemma \ref{lemma:classical strict inequality} cannot be weakened to 
$\rho^0\le\sigma^0$ in general. Indeed, if 
$\sigma=\sum_{\omega\in\Omega}\sigma(\omega)\pr{\omega}$ and 
$\rho=\kappa\sum_{\omega\in\Omega_0}\sigma(\omega)\pr{\omega}$
are commuting states, where $\Omega_0\subseteq\Omega$ and $\kappa\in(0,+\infty)$, then 
it is easy to see that 
$\frac{1}{n}D_{\alpha}^{\test}(\rho^{\otimes n}\|\sigma^{\otimes n})=
D_{\alpha}(\rho\|\sigma)$ for every $n\in\bN$ and $\alpha\in(0,+\infty)$. 
\end{rem}


\begin{thm}\label{thm:an2 main}
Let $\rho$ and $\sigma$ be unequal states 
with $\rho^0=\sigma^0$, such that they are commuting,
and hence can be jointly diagonalized as
$\rho=\sum_{\omega\in\Omega}\rho(\omega)\pr{\omega}$ and 
$\sigma=\sum_{\omega\in\Omega}\sigma(\omega)\pr{\omega}$ 
as in \eqref{commuting pair}.
Then the following are equivalent:
\begin{enumerate}
\item\label{Dhat equality1}
$\hat D_{\alpha}^{\test}(\rho\|\sigma)=D_{\alpha}(\rho\|\sigma)$
for every/some $\alpha\in(0,1)$;
\item\label{Dhat equality2}
$D_{\alpha}^{\test}(\rho\|\sigma)=D_{\alpha}(\rho\|\sigma)$
for every/some $\alpha\in(0,1)$;
\item\label{Dhat equality3}
there exist some $\Omega_0\subseteq\Omega$ and strictly positive constants 
$c_0$, $c_1$ such that 
\begin{align*}
\rho(\omega)=c_0\sigma(\omega),\ds\ds\omega\in\Omega_0,\ds\ds\ds
\rho(\omega)=c_1\sigma(\omega),\ds\ds\omega\in\Omega\setminus\Omega_0.
\end{align*}
\end{enumerate}
In particular, if \ref{Dhat equality3} is not satisfied then $\hat D_{\alpha}^{\test}(\rho\|\sigma)<D_{\alpha}(\rho\|\sigma)$ for every $\alpha\in(0,1)$.
\end{thm}
\begin{proof}
The equivalence \ref{Dhat equality2}$\iff$\ref{Dhat equality3} follows from Corollary 
\ref{cor:classical equality},
and the implication \ref{Dhat equality2}$\imp$\ref{Dhat equality1} is obvious
from the inequalities in Lemma \ref{lemma:Renyi order}.

We prove \ref{Dhat equality1}$\imp$\ref{Dhat equality3} by contraposition. 
Assume that \ref{Dhat equality3} does not hold. Then, by 
Corollary \ref{cor:classical equality}, we have 
$\frac{1}{n}D_{\alpha}^{\test}(\rho^{\otimes n}\|\sigma^{\otimes n})<
\oll{D}_{\alpha}^{\meas}(\rho\|\sigma)=D_{\alpha}(\rho\|\sigma)$ for 
every $\alpha\in(0,1)$ and
$n=1$, 
and the same holds for every $n\ge 2$ by Lemma \ref{lemma:classical strict inequality}. 
Hence, \ref{maximal testdiv smaller2} in Lemma \ref{lemma:maximal testdiv smaller} holds
for every $\alpha\in(0,1)$.
According to Lemma \ref{lemma:Dalpha monotone}, the assumptions that 
$\rho^0=\sigma^0$ and $\rho\ne\sigma$ guarantee that 
$\beta\mapsto D_{\beta}(\rho\|\sigma)$ is strictly increasing on $(0,1)$, whence,
by Theorem \ref{thm:testdiv standard bounds}, 
$\oll{D}_{\alpha}^{\test}(\rho\|\sigma)<D_{\alpha}(\rho\|\sigma)$, i.e., 
\ref{maximal testdiv smaller1} in Lemma \ref{lemma:maximal testdiv smaller} holds
for every $\alpha\in(0,1)$.
Thus, by Lemma \ref{lemma:maximal testdiv smaller},
$\hat D_{\alpha}^{\test}(\rho\|\sigma)< D_{\alpha}(\rho\|\sigma)$
for every $\alpha\in(0,1)$.
\end{proof}

\begin{rem}
Note that \ref{Dhat equality3} in Theorem \ref{thm:an2 main} is a single-copy condition, as are 
$\rho^0=\sigma^0$ and $\rho\ne\sigma$. Hence, these can be easily verified. 

Moreover, 
if $\rho$ and $\sigma$ are 
probability density functions
on a finite set $\Omega$ with $|\Omega|\ge 3$ that are selected randomly according to the uniform distribution 
on the probability simplex, then with probability $1$ the conditions
$\rho^0=\sigma^0$ and $\rho\ne\sigma$ are satisfied, while 
\ref{Dhat equality3} in Theorem \ref{thm:an2 main} fails, and hence
$\oll{D}_{\alpha}^{\test}(\rho\|\sigma)<D_{\alpha}(\rho\|\sigma)$
with probability $1$.
\end{rem}

\section{Conclusion}

We have introduced two different definitions of the regularized test-measured R\'enyi 
$\alpha$-divergences, which coincide with each other and also with the regularized measured R\'enyi 
$\alpha$-divergences, for $\alpha>1$. Our main result is that for $\alpha\in(0,1)$, both versions 
are strictly smaller than the unique classical R\'enyi $\alpha$-divergence for generic pairs of 
classical probability distributions (on at least three points), and hence neither definitions give a quantum extension 
of the classical R\'enyi $\alpha$-divergence. 

In Theorem \ref{thm:different reg} we showed that strict inequality between the two different 
regularized test-measured R\'enyi $\alpha$-divergences may hold as 
$\oll{D}_{\alpha}^{\test}(\rho\|\sigma)<
\hat D_{\alpha}^{\test}(\rho\|\sigma)=
\oll{D}_{\alpha}^{\meas}(\rho\|\sigma)$ for commuting pairs of states. 
It would be interesting to find explicit examples with
non-commuting states where this holds. It is also an open question whether examples exist where the above can be strengthened to 
$\oll{D}_{\alpha}^{\test}(\rho\|\sigma)<
\hat D_{\alpha}^{\test}(\rho\|\sigma)<
\oll{D}_{\alpha}^{\meas}(\rho\|\sigma)$; this would be interesting both in the commuting and 
in the non-commuting cases.

It is also an interesting question whether closed-form expressions can be found for the 
regularized test-measured R\'enyi $\alpha$-divergences for $\alpha\in(0,1)$, similarly to 
the $\alpha>1$ case. 
In this respect probably the best one can hope for is the representation in 
\eqref{alt repr4} based on the results of \cite{Salzmann_Datta21}.

\section*{Acknowledgments}

The work of M.M. was partially funded by the
National Research, Development and 
Innovation Office of Hungary via the research grants K124152 and KH129601, and
by the Ministry of Innovation and
Technology and the National Research, Development and Innovation
Office within the Quantum Information National Laboratory of Hungary.
The authors are indebted to an anonymous reviewer for pointing out the example demonstrating that
$\oll{D}_{\alpha}^\test(\rho\|\sigma)<\hat D_{\alpha}^\test(\rho\|\sigma)$ can happen.

\appendix

\bigskip

\section{Relation to the hypothesis testing problem of Salzmann and Datta}
\label{sec:SD}

It is easy to see that Theorem \ref{thm:testdiv} is closely related to a problem recently studied  
in \cite{Salzmann_Datta21}. We have the following:

\begin{prop}\label{lemma:alt repr}
For any $\alpha\in(0,1)$, 
\begin{align}
\limsup_{n\to+\infty}\frac{1}{n}D_{\alpha}^{\test}(\rho^{\otimes n}\|\sigma^{\otimes n})
&=
\limsup_{n\to+\infty}-\frac{1}{n}\log\min_{T\in\bT(\hil^{\otimes n})}\left\{
\bz\Tr\rho^{\otimes n}(I-T)\jz^{\frac{\alpha}{1-\alpha}}+\Tr\sigma^{\otimes n}T\right\},
\label{alt repr1}\\
\liminf_{n\to+\infty}\frac{1}{n}D_{\alpha}^{\test}(\rho^{\otimes n}\|\sigma^{\otimes n})
&=
\liminf_{n\to+\infty}-\frac{1}{n}\log\min_{T\in\bT(\hil^{\otimes n})}\left\{
\bz\Tr\rho^{\otimes n}(I-T)\jz^{\frac{\alpha}{1-\alpha}}+\Tr\sigma^{\otimes n}T\right\}.
\label{alt repr2}
\end{align}
In particular, 
\begin{align}
\exists \s \lim_{n\to+\infty}\frac{1}{n}D_{\alpha}^{\test}(\rho^{\otimes n}\|\sigma^{\otimes n})\ds\iff\ds
\exists\s\lim_{n\to+\infty}-\frac{1}{n}\log\min_{T\in\bT(\hil^{\otimes n})}\left\{
\bz\Tr\rho^{\otimes n}(I-T)\jz^{\frac{\alpha}{1-\alpha}}+\Tr\sigma^{\otimes n}T\right\},
\label{alt repr3}
\end{align}
and if the limits exist then they are equal.
\end{prop}
\begin{proof}
We prove the equality in \eqref{alt repr1}, as the equality in \eqref{alt repr2} follows in a similar way.

For every $n\in\bN$, let $T_n$ be the minimizer of 
$\bz\Tr\rho^{\otimes n}(I-T)\jz^{\frac{\alpha}{1-\alpha}}+\Tr\sigma^{\otimes n}T$
over all tests, and let 
$c<\limsup_{n\to+\infty}-\frac{1}{n}\log\bz\bz\Tr\rho^{\otimes n}(I-T_n)\jz^{\frac{\alpha}{1-\alpha}}+\Tr\sigma^{\otimes n}T_n\jz$. Then there exists a strictly increasing sequence
$(n_k)_{k\in\bN}$ in $\bN$ such that 
$\bz\Tr\rho^{\otimes n_k}(I-T_{n_k})\jz^{\frac{\alpha}{1-\alpha}}<e^{-n_kc}$, 
$\Tr\sigma^{\otimes n_k}T_{n_k}<e^{-n_kc}$ for all $k\in\bN$, whence
\begin{align*}
\underbrace{\bz\Tr\rho^{\otimes n_k}T_{n_k}\jz^{\alpha}}_{\le 1}
\underbrace{\bz \Tr\sigma^{\otimes n_k}T_{n_k}\jz^{1-\alpha}}_{<e^{-n_kc(1-\alpha)}}
+
\underbrace{\bz\Tr\rho^{\otimes n_k}(I-T_{n_k})\jz^{\alpha}}_{<e^{-n_kc(1-\alpha)}}
\underbrace{\bz\Tr\sigma^{\otimes n_k}(I-T_{n_k})\jz^{1-\alpha}}_{\le 1}
<
2e^{-n_kc(1-\alpha)}.
\end{align*}
Thus, 
\begin{align*}
\limsup_{n\to+\infty}\frac{1}{n}D_{\alpha}^{\test}(\rho^{\otimes n}\|\sigma^{\otimes n})
\ge
\limsup_{k\to+\infty}\frac{1}{n_k}D_{\alpha}\bz\T_{n_k}(\rho^{\otimes n_k})\|\T_{n_k}(\rho^{\sigma n_k})\jz
\ge c. 
\end{align*}
Since this holds for all $c$ as above, we get LHS$\ge$RHS in \eqref{alt repr1}.

If the LHS in \eqref{alt repr1} is $0$ then,
by the above, both expressions in \eqref{alt repr1} are $0$, and the proof is complete. Hence, for the rest we assume that the LHS in \eqref{alt repr1} is strictly positive.

For every $n\in\bN$, let $T_n$ be a test achieving $D_{\alpha}^{\test}(\rho^{\otimes n}\|\sigma^{\otimes n})$, and let 
$0<c<\limsup_{n\to+\infty}\frac{1}{n}D_{\alpha}\bz\T_n(\rho^{\otimes n})\|\T_n(\sigma^{\otimes n})\jz$.
Then there exists a strictly increasing sequence
$(n_k)_{k\in\bN}$ in $\bN$ such that for every $k\in\bN$,
\begin{align}
e^{-n_k(1-\alpha)c}
>
(\Tr\rho^{\otimes n_k}T_{n_k})^{\alpha}(\Tr\sigma^{\otimes n_k}T_{n_k})^{1-\alpha}
+
(\Tr\rho^{\otimes n_k}(I-T_{n_k}))^{\alpha}(\Tr\sigma^{\otimes n_k}(I-T_{n_k}))^{1-\alpha}\nonumber\\
\ge
\left[
(\Tr\rho^{\otimes n_k}T_{n_k})^{\frac{\alpha}{1-\alpha}}\Tr\sigma^{\otimes n_k}T_{n_k}
+
(\Tr\rho^{\otimes n_k}(I-T_{n_k}))^{\frac{\alpha}{1-\alpha}}\Tr\sigma^{\otimes n_k}(I-T_{n_k})\right]
^{1-\alpha}.\label{DS proof1}
\end{align}
Since $D_{\alpha}\bz\T_{n_k}(\rho^{\otimes n})\|\T_{n_k}(\sigma^{\otimes n})\jz$, as well as the expression in \eqref{DS proof1}, are symmetric under exchanging $T_{n_k}$ with 
$I-T_{n_k}$, we may assume that 
$\Tr\sigma^{\otimes n_k}T_{n_k}\le 1/2$, and thus 
$\Tr\sigma^{\otimes n_k}(I-T_{n_k})\ge 1/2$ for every $k\in\bN$.
This, together with \eqref{DS proof1}, implies that 
\begin{align*}
(\Tr\rho^{\otimes n_k}(I-T_{n_k}))^{\frac{\alpha}{1-\alpha}}<2e^{-n_kc},\ds\ds\ds
k\in\bN.
\end{align*}
In particular, 
$\lim_{k\to+\infty}\Tr\rho^{\otimes n_k}T_{n_k}=1$, so that
$\Tr\rho^{\otimes n_k}T_{n_k}>(1/2)^{\frac{1-\alpha}{\alpha}}$ for every $k\ge k_0$
with some $k_0\in\bN$, which, when combined with 
\eqref{DS proof1}, yields 
\begin{align*}
\Tr\sigma^{\otimes n_k}T_{n_k}<2e^{-n_kc},\ds\ds\ds k\ge k_0.
\end{align*}
Thus,
\begin{align*}
&\limsup_{n\to+\infty}-\frac{1}{n}\log\min_{0\le T\le I}\left\{
\bz\Tr\rho^{\otimes n}(I-T)\jz^{\frac{\alpha}{1-\alpha}}+\Tr\sigma^{\otimes n}T\right\}\\
&\ds\ge
\limsup_{k\to+\infty}-\frac{1}{n_k}\log
\bz (\Tr\rho^{\otimes n_k}(I-T_{n_k}))^{\frac{\alpha}{1-\alpha}}
+
\Tr\sigma^{\otimes n_k}T_{n_k}\jz
\ge c.
\end{align*}
Since this holds for every $c$ as above, we get 
that LHS$\le$RHS in \eqref{alt repr1}.
\end{proof}

\medskip
Theorem \ref{thm:testdiv} shows that the limit on the LHS in 
\eqref{alt repr3} exists, and is equal to the expressions in 
\eqref{testdiv formula}, while \cite[Theorem 7]{Salzmann_Datta21}
shows that the limit on the RHS in \eqref{alt repr3} exists, and gives a different expression for it. Combining the two results yields that 
\begin{align}
\oll{D}_{\alpha}^{\test}(\rho\|\sigma)
&=
\lim_{n\to+\infty}\frac{1}{n}D_{\alpha}^{\test}(\rho^{\otimes n}\|\sigma^{\otimes n})
=\sup\left\{r:\,\frac{H_r}{r}\ge \frac{1-\alpha}{\alpha} \right\}\nonumber\\
&=
\lim_{n\to+\infty}-\frac{1}{n}\log\min_{0\le T\le I}\left\{
\bz\Tr\rho^{\otimes n}(I-T)\jz^{\frac{\alpha}{1-\alpha}}+\Tr\sigma^{\otimes n}T\right\}
\nonumber\\
&=
\alpha\sup_{0< t< 1}\frac{(t-1)D_t(\rho\|\sigma)}{t(2\alpha-1)-\alpha},\nn
\end{align}
where the last equality was proved in \cite{Salzmann_Datta21}.

\section{The von Neumann algebra case}
\label{sec:vN}

In this appendix we extend 
the main results in Section \ref{sec:an1}
to the general von Neumann algebra case.
For that, it is sufficient to extend Theorem \ref{thm:testdiv};
once that is done, the extensions of
Theorem \ref{thm:testdiv standard bounds} and Corollaries
\ref{cor:regularized strictly positive}--\ref{C-III.29} 
follow the same way as in Section \ref{sec:an1}.
Moreover, the essential thing for 
the proof of Theorem \ref{thm:testdiv} is the version of the 
Hoeffding bound theorem given in Lemma \ref{lemma:Hoeffding}, 
so it is crucial for our purpose to extend Lemma \ref{lemma:Hoeffding}
to the von Neumann algebra setting, which was mainly done in 
\cite{JOPS}, under some technical assumptions. Our main contribution below is 
removing those assumptions, and showing that Lemma \ref{lemma:Hoeffding}
can be extended to the von Neumann algebra setting for arbitrary 
pairs of normal states $\rho,\sigma$.

We assume that the reader is familiar with the standard form of von Neumann algebras and the notions
of the relative modular operator and Connes' cocycle derivative. We refer to \cite{Hiai_Lectures2021} for the details.

Let $\M$
be a von Neumann algebra represented in the \emph{standard form} $(\M,\hil,J,\P)$. Let
$\rho,\sigma\in\M_*^+$ (normal positive linear functionals on $\M$) and $\Delta_{\rho,\sigma}$ be
the \emph{relative modular operator} $\Delta_{\rho,\sigma}$ with the spectral decomposition
$\Delta_{\rho,\sigma}=\int_0^{+\infty}t\,dE_{\rho,\sigma}(t)$. 
We define
\[
Q_\alpha(\rho\|\sigma):=\|\Delta_{\rho,\sigma}^{\alpha/2}\xi_\sigma\|^2
=\int_{(0,+\infty)}t^\alpha\,d\|E_{\rho,\sigma}(t)\xi_\sigma\|^2,
\qquad\alpha\in[0,1],
\]
where $\xi_\sigma\in\P$ is the vector representative of $\sigma$ so that
$\sigma(x)=\langle\xi_\sigma,x\xi_\sigma\rangle$ ($x\in\M$). We define
\[
\psi(\alpha)=\psi(\alpha|\rho\|\sigma):=\log Q_\alpha(\rho\|\sigma),\qquad\alpha\in[0,1],
\]
and further define $\vfi(c)$ and $\vfi_+(c)$ for $c\in\bR$ as in \eqref{phi def}--\eqref{phi+ def},
where the existence of the maxima in the definitions of $\vfi(c),\vfi_+(c)$ is clear from the continuity
of $\psi$ (see \ref{Leg1} in Lemma \ref{L-C.1}). In our discussions below it is convenient to consider
the measure $d\nu_{\rho,\sigma}(t):=d\|E_{\rho,\sigma}(t)\xi_\sigma\|^2$ for $t\in(0,+\infty)$, i.e.,
$\nu_{\rho,\sigma}$ is the spectral measure of $\Delta_{\rho,\sigma}$ on
$s(\Delta_{\rho,\sigma})\hil$ 
with respect to $\xi_\sigma$, where $s(\Delta_{\rho,\sigma})$ is the
support projection of $\Delta_{\rho,\sigma}$. 
Note that $s(\Delta_{\rho,\sigma})=s(\rho)Js(\sigma)J$; see \cite[Proposition 10.3]{Hiai_Lectures2021}.
We then have
$Q_\alpha(\rho\|\sigma)=\int_{(0,+\infty)}t^\alpha\,d\nu_{\rho,\sigma}(t)$ for all $\alpha\in[0,1]$.

Let $\rho,\sigma\in\M_*^+$ be normal states. The \emph{standard $\alpha$-R\'enyi divergence}
$D_\alpha(\rho\|\sigma)$ with parameter $\alpha\in[0,1)$ is
\[
D_\alpha(\rho\|\sigma):={1\over\alpha-1}\log Q_\alpha(\rho\|\sigma).
\]
Properties of $D_\alpha(\rho\|\sigma)$ as well as $Q_\alpha(\rho\|\sigma)$ in the present
setting can be found in \cite{Hiai_fdiv_standard}.
The \emph{Hoeffding divergence} $H_r(\rho\|\sigma)$ of $\rho,\sigma$ for $r\in\bR$ is
defined by
\begin{align}\label{Hr-def}
H_r(\rho\|\sigma):=\sup_{\alpha\in(0,1)}{(\alpha-1)r-\psi(\alpha)\over\alpha}
=\sup_{\alpha\in(0,1)}{\alpha-1\over\alpha}\left[r-D_\alpha(\rho\|\sigma)\right].
\end{align}

Just like in Section \ref{sec:def}, the \emph{test-measured R\'enyi $\alpha$-divergence} of
$\rho,\sigma$ is defined as
\[
D_\alpha^\test(\rho\|\sigma):=\sup_{T\in\M,\,0\le T\le1}D_\alpha(\T(\rho)\|\T(\sigma))
=\max_{T\in\M,\,0\le T\le1}D_\alpha(\T(\rho)\|\T(\sigma)),
\]
where $\T(\rho):=(\rho(T),\rho(1-T))$ and similarly for $\T(\sigma)$. The \emph{regularized
test-measured R\'enyi $\alpha$-divergence} is defined as
\[
\oll{D}_\alpha^\test(\rho\|\sigma)
:=\limsup_{n\to\infty}{1\over n}D_\alpha^\test(\rho^{\otimes n}\|\sigma^{\otimes n}),
\]
where the $n$-fold tensor products $\rho^{\otimes n},\sigma^{\otimes n}$ are normal states on the
$n$-fold von Neumann algebra tensor product $\M^{\overline\otimes n}$.

If $s(\rho)\perp s(\sigma)$ for the support projections of $\rho,\sigma$, then it is immediate to see
that $D_\alpha(\rho\|\sigma)=D_\alpha(\rho^{\otimes n}\|\sigma^{\otimes n})=+\infty$ for all
$\alpha\in[0,1)$ and $n\in\bN$, and $H_r(\rho\|\sigma)=+\infty$ for all $r\in\bR$; thus Theorem
\ref{thm:testdiv} holds trivially. So in the rest we always assume that $s(\rho)\not\perp s(\sigma)$.

\begin{lemma}\label{L-C.1}
Let $\rho,\sigma$ be normal states on $\M$.
\begin{enumerate}
\item \label{Leg1}
$\psi$ is a $(-\infty,0]$-valued continuous and convex function on $[0,1]$,
and real analytic in $(0,1)$.
\item \label{Leg2}
$\vfi_+$ is a strictly increasing function on $(\psi'(0^+),+\infty)$ mapping onto $(-\psi(0),+\infty)$.
\item \label{Leg3}
For every $r>-\psi(0)$ there exists a unique $c_r\in(\psi'(0^+),+\infty)$ such that $\vfi_+(c_r)=r$ and
$H_r(\rho\|\sigma)=r-c_r=\vfi(c_r)$.
\item \label{Leg4}
$r\mapsto H_r(\rho\|\sigma)$ is convex, lower semi-continuous, and monotone
decreasing on $\bR$, and for every $r<-\psi(0)$, $H_r(\rho\|\sigma)=+\infty$.
\end{enumerate}
\end{lemma}

\begin{proof}
\ref{Leg1}:\enspace
That $Q_\alpha(\rho\|\sigma)>0$ for all $\alpha\in[0,1]$ follows from the assumption
$s(\rho)\not\perp s(\sigma)$. As in the proof of \cite[Proposition 5.3]{Hiai_fdiv_standard}, convexity of
$\psi$ is a consequence of the H\"older inequality
\begin{align}\label{F-C.1}
\int_{(0,+\infty)}t^{\lambda\alpha_1+(1-\lambda)\alpha_2}\,d\nu_{\rho,\sigma}(t)
\le\biggl[\int_{(0,+\infty)}t^{\alpha_1}\,d\nu_{\rho,\sigma}(t)\biggr]^\lambda
\biggl[\int_{(0,+\infty)}t^{\alpha_2}\,d\nu_{\rho,\sigma}(t)\biggr]^{1-\lambda}.
\end{align}
Since
\begin{align*}
Q_0(\rho\|\sigma)&=\|s(\Delta_{\rho,\sigma})\xi_\sigma\|^2
=\|s(\rho)\xi_\sigma\|^2=\sigma(s(\rho))\le1, \\
Q_1(\rho\|\sigma)&=\|\Delta_{\rho,\sigma}^{1/2}\xi_\sigma\|^2
=\|s(\sigma)\xi_\rho\|^2=\rho(s(\sigma))\le1,
\end{align*}
convexity of $\psi$ implies that $\psi(\alpha)\le0$ for all $\alpha\in[0,1]$. Note (see, e.g.,
\cite[Theorem A.7]{Hiai_Lectures2021}) that $\Delta_{\rho,\sigma}^{z/2}\xi_\sigma$ is continuous
on $0\le\mathrm{Re}\,z\le1$ and analytic in $0<\mathrm{Re}\,z<1$ in the strong operator topology,
so that $\bigl\langle\Delta_{\rho,\sigma}^{\overline z/2}\xi_\sigma,
\Delta_{\rho,\sigma}^{z/2}\xi_\sigma\bigr\rangle$ is analytic in $0<\mathrm{Re}\,z<1$. Hence
$Q_\alpha(\rho\|\sigma)$ is real analytic in $0<\alpha<1$ and so is $\psi(\alpha)$.

The proofs of the remaining \ref{Leg2}--\ref{Leg4} are similar to those in \cite[Sec.~IV]{HMO2} in the
finite-dimensional case, while we give them for readers' convenience.

\ref{Leg2}:\enspace
It is obvious that $\vfi_+$ is strictly increasing on $[\psi'(1^-),+\infty)$
(whenever $\psi'(1^-)<+\infty$). Let us show that $\vfi_+$ is strictly increasing on
$(\psi'(0^+),\psi'(1^-))$ (whenever $\psi'(0^+)<\psi'(1^-)$). By \ref{Leg1} note that $\psi'(\alpha)$ is strictly
increasing in $(0,1)$. For every $a\in(\psi'(0^+),\psi'(1^-))$ there exists a unique $\alpha_a\in(0,1)$
such that $\psi'(\alpha_a)=a$ and hence $\vfi_+(a)=a\alpha_a-\psi(\alpha_a)$. Let
$\psi'(0^+)<a<b<\psi'(1^-)$. Then, since
\[
a\alpha_a-\vfi_+(a)=\psi(\alpha_a)>b(\alpha_a-\alpha_b)+\psi(\alpha_b)
=b\alpha_a-\vfi_+(b),
\]
one has $\vfi_+(b)-\vfi_+(a)>(b-a)\alpha_a>0$, so $\vfi_+(a)<\vfi_+(b)$. Hence $\vfi_+$ is strictly
increasing on $(\psi'(0^+),+\infty)$. Furthermore, it is immediate to see that
$\vfi_+(a)\to+\infty$ as $a\to+\infty$ and $\vfi_+(a)\to-\psi(0)$ as $a\to\psi'(0^+)$.

\ref{Leg3}:\enspace
Let $r>-\psi(0)$. By (ii) there exists a unique $c_r\in(\psi'(0^+),+\infty)$ such that
$\vfi_+(c_r)=r$. Then, in view of \eqref{phi+ def}, $\psi(\alpha)\ge c_r\alpha-r$ for all $\alpha\in[0,1]$
and $\psi(\alpha_r)=c_r\alpha_r-r$ for some $\alpha_r\in[0,1]$. Since
$r>-\psi(0)$ and hence $c_r>\psi'(0^+)$, one has $\alpha_r\in(0,1]$ and
$(r+\psi(\alpha_r))/\alpha_r=c_r$. Therefore, 
\[
H_r(\rho\|\sigma)=\sup_{\alpha\in(0,1]}{(\alpha-1)r-\psi(\alpha)\over\alpha}
=r-c_r=\vfi_+(c_r)-c_r=\vfi(c_r).
\]

\ref{Leg4}:\enspace
Since $\alpha\mapsto D_\alpha(\rho\|\sigma)$ is continuous on $(0,1)$ (see the proof of \ref{Leg1}),
the assertion on $r\mapsto H_r(\rho\|\sigma)$ is obvious by \eqref{Hr-def}.
If $r<-\psi(0)$, then $-r-\psi(\alpha)\to-r-\psi(0)>0$ as $\alpha\to0^+$, so it is immediate that
$H_r(\rho\|\sigma)=+\infty$.
\end{proof}

\begin{lemma}\label{L-C.2}
For every normal states $\rho,\sigma$ on $\M$ and any $n\in\bN$,
\[
\psi(\alpha|\rho^{\otimes n}\|\sigma^{\otimes n})=n\psi(\alpha|\rho\|\sigma),\quad\alpha\in[0,1];
\qquad
D_\alpha(\rho^{\otimes n}\|\sigma^{\otimes n})=nD_\alpha(\rho\|\sigma),\quad\alpha\in[0,1).
\]
\end{lemma}

\begin{proof}
It suffices to show that when $\M_k$ ($k=1,2$) are von Neumann algebras and
$\rho_k,\sigma_k\in(\M_k)_*^+$, we have
\begin{align}\label{Q-tensor}
Q_\alpha(\rho_1\otimes\rho_2\|\sigma_1\otimes\sigma_2)
=Q_\alpha(\rho_1\|\sigma_1)Q_\alpha(\rho_2\|\sigma_2),\qquad\alpha\in[0,1].
\end{align}
To prove this let $M_k$ be represented in the standard form
$(\M_k,\hil_k,J_k,\P_k)$, $k=1,2$. Then the standard form of
$\M_{12}:=\M_1\overline\otimes\M_2$ is given as $(\M_{12},\hil_{12},J_{12},\P_{12})$ with
$\hil_{12}:=\hil_1\otimes\hil_2$, $J_{12}:=J_1\otimes J_2$ and
$\P_{12}\supseteq\{\xi_1\otimes\xi_2:\xi_k\in\P_k,\,k=1,2\}$. Let $\xi_{\rho_k},\xi_{\sigma_k}\in\P_k$
be the vector representatives of $\rho_k,\sigma_k$, so that
$\xi_{\rho_1}\otimes\xi_{\rho_2},\xi_{\sigma_1}\otimes\xi_{\sigma_2}\in\P_{12}$ are those of
$\rho_1\otimes\rho_2,\sigma_1\otimes\sigma_2$ respectively. With use of Connes' cocycle derivative
$(D\rho:D\sigma)_t$ for $\rho,\sigma\in\M_*^+$ (see \cite{Hiai_Lectures2021,Stratila-book})
note by \cite[Proposition 10.11]{Hiai_Lectures2021} that
\begin{align*}
\Delta_{\rho_1\otimes\rho_2,\sigma_1\otimes\sigma_2}^{it}
(\xi_{\sigma_1}\otimes\xi_{\sigma_2})
&=(D(\rho_1\otimes\rho_2):D(\sigma_1\otimes\sigma_2))_t
(\xi_{\sigma_1}\otimes\xi_{\sigma_2}) \\
&=((D\rho_1:D\sigma_1)_t\otimes(D\rho_2:D\sigma_2)_t)
(\xi_{\sigma_1}\otimes\xi_{\sigma_2}) \\
&=(D\rho_1:D\sigma_1)_t\xi_{\sigma_1}\otimes(D\rho_2:D\sigma_2)_t\xi_{\sigma_2} \\
&=\Delta_{\rho_1,\sigma_1}^{it}\xi_{\sigma_1}\otimes
\Delta_{\rho_2,\sigma_2}^{it}\xi_{\sigma_2},\qquad t\in\bR,
\end{align*}
where the equality above is due to \cite[Sec.~3.9]{Stratila-book}, which is given
under the assumption $s(\rho_1)\le s(\sigma_1)$ and $s(\rho_2)\le s(\sigma_2)$ but it can be removed.
By analytic continuation and continuity we have
\[
\Delta_{\rho_1\otimes\rho_2,\sigma_1\otimes\sigma_2}^{\alpha/2}
(\xi_{\sigma_1}\otimes\xi_{\sigma_2})
=\Delta_{\rho_1,\sigma_1}^{\alpha/2}\xi_{\sigma_1}\otimes
\Delta_{\rho_2,\sigma_2}^{\alpha/2}\xi_{\sigma_2},\qquad\alpha\in(0,1),
\]
which gives \eqref{Q-tensor}.
\end{proof}

We write for every $\rho,\sigma\in\M_*^+$,
\[
\chi(\rho\|\sigma):=\min_{T\in\M,\,0\le T\le1}\{\rho(1-T)+\sigma(T)\}
={1\over2}\{\rho(1)+\sigma(1)-\|\rho-\sigma\|\},
\]
where the minimum is attained by $T=s((\rho-\sigma)_+)$, the support projection of the positive part
of $\rho-\sigma$. The Chernoff bound theorem in \cite{Aud,NSz} (see also \cite{HMO2}) was extended
in \cite[Theorem 6.5]{JOPS} to the non-i.i.d.\ von Neumann algebra setting. In our
i.i.d.\ setting \cite[Theorem 6.5]{JOPS} says the following:
\begin{lemma}\label{lemma:JOPS-Chernoff}
For any pair of normal states $\rho,\sigma\in\M_*^{+}$ such that 
$\psi'(0^+)<\psi'(1^-)$, (i.e., $\psi$ is not affine on $[0,1]$), we have 
\begin{align}\label{F-C.2}
\lim_{n\to\infty}{1\over n}\log\chi(\rho^{\otimes n}\|\sigma^{\otimes n})=-\vfi(0).
\end{align}
\end{lemma}
\begin{rem}
The equality in \eqref{F-C.2} was proved in \cite[Theorem 6.5]{JOPS} under the assumption that 
$\rho,\sigma$ are faithful; however, it is easy to see that this assumption 
is not needed in the proof. 
\end{rem}

Keys to prove \eqref{F-C.2} are the following inequalities for $\rho,\sigma\in\M_*^+$
\cite[Theorem 6.1]{JOPS}:
\begin{align}
\chi(\rho\|\sigma)
&\le Q_\alpha(\rho\|\sigma)=\int_{(0,+\infty)}t^\alpha\,d\nu_{\rho,\sigma}(t),
\qquad\alpha\in[0,1], \label{F-C.3}\\
\chi(\rho\|\sigma)
&\ge\langle \xi_\sigma,\Delta_{\rho,\sigma}(1+\Delta_{\rho,\sigma})^{-1}\xi_\sigma\rangle
=\int_{(0,+\infty)}{t\over1+t}\,d\nu_{\rho,\sigma}(t). \label{F-C.4}
\end{align}
In Lemma \ref{L-C.5} below we extend \eqref{F-C.2} to the form \eqref{Nagaoka bound} with removing
the assumption $\psi'(0^+)<\psi'(1^-)$. To do so, we give two more lemmas.

\begin{lemma}\label{L-C.3}
For every $\rho,\sigma\in\M_*^+$ the following conditions are equivalent:
\begin{enumerate}
\item \label{affine1}
$\psi$ is affine on $[0,1]$, i.e., $\psi'(0^+)=\psi'(1^-)$;
\item \label{affine2}
$\nu_{\rho,\sigma}$ is supported at a single point $\kappa\in(0,+\infty)$.
\end{enumerate}
In this case, $\psi(\alpha)=\alpha\log\kappa+\log\sigma(s(\rho))$ for $\alpha\in[0,1]$, and
$\nu_{\rho,\sigma}=\sigma(s(\rho))\delta_\kappa$.
\end{lemma}

\begin{proof}
Assume \ref{affine1}; then inequality \eqref{F-C.1} must always be an equality. From the equality case
of the H\"older inequality one can easily see that $\nu_{\rho,\sigma}$ is supported at a single point.
Conversely, assume \ref{affine2} so that $\nu_{\rho,\sigma}=\sigma(s(\rho))\delta_\kappa$ with some
$\kappa\in(0,+\infty)$, since $\nu_{\rho,\sigma}((0,+\infty))=\sigma(s(\rho))$. Then
\[
\psi(\alpha)=\log[\kappa^\alpha \sigma(s(\rho))]=\alpha\log\kappa+\log\sigma(s(\rho)),
\qquad\alpha\in[0,1],
\]
and hence \ref{affine1} follows.
\end{proof}
We give some further characterizations of 
$\psi$ being affine
in Remark \ref{R-C.4} below.
\medskip

The equality in \eqref{F-C.5} below is an extension of \eqref{F-C.2}, and it has been
shown in the proof of \cite[Theorem 6.6]{JOPS} for faithful normal states and
$b\in(\psi'(0^+),\psi'(1^-))$. We show that it holds for arbitrary pairs of normal states and $b$ values.

\begin{lemma}\label{L-C.5}
For every normal states $\rho,\sigma\in\M_*^+$ and any $b\in\bR$,
\begin{align}\label{F-C.5}
\lim_{n\to+\infty}{1\over n}\log\chi(\rho^{\otimes n}\|e^{nb}\sigma^{\otimes n})=-\vfi(b).
\end{align}
\end{lemma}

\begin{proof}
First, we show that \eqref{F-C.2} holds also when $\psi$ is affine on $[0,1]$. Since the inequality
\[
\limsup_{n\to+\infty}{1\over n}\log\chi(\rho^{\otimes n}\|\sigma^{\otimes n})\le-\vfi(0)
\]
was shown in \cite{JOPS} from \eqref{F-C.3} for general $\rho,\sigma$, we need only to show that
\begin{align}\label{F-C.6}
\liminf_{n\to+\infty}{1\over n}\log\chi(\rho^{\otimes n}\|\sigma^{\otimes n})\ge-\vfi(0).
\end{align}
When $\psi$ is affine, by Lemma \ref{L-C.3} one has
$\nu_{\rho,\sigma}=\sigma(s(\rho))\delta_\kappa$ and
$\psi(\alpha)=\alpha\log\kappa+\log\sigma(s(\rho))$ for $\alpha\in[0,1]$. Therefore,
\begin{align}\label{F-C.7}
-\vfi(0)=\begin{cases}\log[\kappa\sigma(s(\rho))] & \text{if $0<\kappa\le1$}, \\
\log\sigma(s(\rho)) & \text{if $\kappa>1$}.\end{cases}
\end{align}
For each $n\in\bN$, since $\psi(\alpha|\rho^{\otimes n}\|\sigma^{\otimes n})
=\alpha\log\kappa^n+\log\sigma(s(\rho))^n$ by Lemma \ref{L-C.2}, it follows from Lemma \ref{L-C.3}
again that $\nu_{\rho^{\otimes n},\sigma^{\otimes n}}=\sigma(s(\rho))^n\delta_{\kappa^n}$. Moreover,
since $t/(1+t)\ge{1\over2}\min\{t,1\}$ for $t\in(0,+\infty)$, by \eqref{F-C.4} for
$\rho^{\otimes n},\sigma^{\otimes n}$ one has
\[
\chi(\rho^{\otimes n}\|\sigma^{\otimes n})
\ge{1\over2}\int_{(0,+\infty)}\min\{t,1\}\,d\mu_{\rho^{\otimes n},\sigma^{\otimes n}}(t)
={1\over2}\min\{\kappa^n,1\}\sigma(s(\rho))^n.
\]
Therefore,
\begin{align}\label{F-C.8}
\liminf_{n\to+\infty}{1\over n}\log\chi(\rho^{\otimes n}\|\sigma^{\otimes n})
\ge\begin{cases}\log[\kappa\sigma(s(\rho)) & \text{if $0<\kappa\le1$}, \\
\log\sigma(s(\rho)) & \text{if $\kappa>1$}.\end{cases}
\end{align}
Combining \eqref{F-C.7} and \eqref{F-C.8} yields \eqref{F-C.6}.
Thus, by the above and Lemma \ref{lemma:JOPS-Chernoff}, we obtain \eqref{F-C.5} with $b=0$ without 
any assumption on the states $\rho,\sigma\in\M_*^{+}$.

To prove \eqref{F-C.5}, we observe that the proof of \eqref{F-C.2} in \cite{JOPS}, as well as the
discussions in this section so far, works even for general $\rho,\sigma\in\M_*^+$
(without being states). Therefore,
\[
\lim_{n\to\infty}{1\over n}\log\kappa_n(\rho^{\otimes n}\|e^{nb}\sigma^{\otimes n})
=-\vfi(0|\rho\|e^b\sigma),
\]
where we write $\vfi(c|\rho\|\sigma)$ for $\vfi(c)$ given $\rho,\sigma$. It is easy to confirm that
$\vfi(0|\rho\|e^b\sigma)=\vfi(b)$.
\end{proof}

Next we extend Lemma \ref{lemma:Hoeffding} to the von Neumann algebra setting.

\begin{lemma}\label{L-C.6}
Let $\rho,\sigma$ be normal states on $\M$. For any $r\in(0,+\infty)$ and $\alpha\in(0,1)$,
\[
\sigma^{\otimes n}(T_{n,r,\alpha})\le e^{-nr},
\qquad\rho^{\otimes n}(1-T_{n,r,\alpha})
\le e^{-n{\alpha-1\over\alpha}[r-D_\alpha(\rho\|\sigma)]},
\]
where $T_{n,r,\alpha}$ is the support projection of
$(\rho^{\otimes n}-e^{n(r+\psi(\alpha))/\alpha}\sigma^{\otimes n})_+$.

Furthermore, for any $r>D_0(\rho\|\sigma)$, any test sequence $(T_n)_{n\in\bN}$ (i.e.,
$T_n\in\M^{\overline\otimes n}$, $0\le T_n\le1$ for $n\in\bN$) and any strictly increasing
sequence $(n_k)_{k\in\bN}$,
\[
\mbox{if}\quad
\liminf_{k\to+\infty}-{1\over n_k}\log\sigma^{\otimes n_k}(T_{n_k})\ge r
\quad\mbox{then}\quad
\limsup_{k\to+\infty}-{1\over n_k}\log\rho^{\otimes n_k}(1-T_{n_k})\le H_r(\rho\|\sigma).
\]
\end{lemma}

\begin{proof}
For every $r\in(0,+\infty)$ and $\alpha\in(0,1)$ let $b:=(r+\psi(\alpha))/\alpha$. Using \eqref{F-C.3}
for $\rho^{\otimes n}$ and $e^{nb}\sigma^{\otimes n}$ implies that
\[
Q_\alpha(\rho^{\otimes n}\|e^{nb}\sigma^{\otimes n})
\ge\chi(\rho^{\otimes n}\|e^{nb}\sigma^{\otimes n})
=\rho^{\otimes n}(1-T_{n,r,\alpha})+e^{nb}\sigma^{\otimes n}(T_{n,r,\alpha}).
\]
Since
\[
Q_\alpha(\rho^{\otimes n}\|e^{nb}\sigma^{\otimes n})
=e^{nb(1-\alpha)}Q_\alpha(\rho^{\otimes n}\|\sigma^{\otimes n})
=e^{n(b(1-\alpha)+\psi(\alpha))},
\]
one has
\begin{align*}
\sigma^{\otimes n}(T_{n,r,\alpha})&\le e^{n(-b\alpha+\psi(\alpha))}=e^{-nr}, \\
\rho^{\otimes n}(1-T_{n,r,\alpha})&\le e^{n(b(1-\alpha)+\psi(\alpha))}
=e^{n[r(1-\alpha)+\psi(\alpha)]/\alpha}
=e^{-n{\alpha-1\over\alpha}[r-D_\alpha(\rho\|\sigma)]}.
\end{align*}

The proof of the latter assertion is the same as that of Lemma \ref{lemma:Hoeffding}, by using
Lemmas \ref{L-C.5} and \ref{L-C.1}.
\end{proof}

We are now in a position to present the main result in this appendix.

\begin{thm}
For any normal states $\rho,\sigma$ on $\M$ and any $\alpha\in(0,1)$, the expressions in Theorem
\ref{thm:testdiv} hold.
\end{thm}

\begin{proof}
Once Lemmas \ref{L-C.1} and \ref{L-C.6} have been shown, the proof of Theorem \ref{thm:testdiv}
remains valid as it is, even in the more general von Neumann algebra setting. Indeed, we may just replace
$\Tr\rho^{\otimes n}T_n$, $\Tr\sigma^{\otimes n}T_n$ with $\rho^{\otimes n}(T_n)$,
$\sigma^{\otimes n}(T_n)$ respectively.
\end{proof}

\begin{rem}
It is quite easy to verify that the proof of \cite[Lemma 4]{Salzmann_Datta21}
works without alteration in the above general von Neumann algebra setting, and 
so does the proof of Proposition \ref{lemma:alt repr} as well, whence we also have 
\eqref{alt repr4}.
\end{rem}

We close this appendix with the following supplement to Lemma \ref{L-C.3}, which might be of independent interest.

\begin{rem}\label{R-C.4}
In the finite-dimensional case, the condition of $\psi$ being affine was explicitly characterized in
\cite[Lemma 3.2]{HMO2} in terms of the density operators $\rho,\sigma$. We can prove that if
$\rho,\sigma\in\M_*^+$ are such that $s(\rho)\le s(\sigma)$, then the following conditions are equivalent:
\begin{enumerate}
\item $\psi(\alpha|\rho\|\sigma)$ is affine on $[0,1]$;
\item $D_0(\rho(1)^{-1}\rho\|\sigma)=D(\rho(1)^{-1}\rho\|\sigma)$;
\item $s(\rho)$ is in the centralizer of $\sigma|_{s(\sigma)\M s(\sigma)}$ (see \cite{Hiai_Lectures2021}
for the definition) and $\rho=\kappa\sigma(s(\rho)\,\cdot)$ for some constant $\kappa>0$.
\end{enumerate}

Note that in the finite-dimensional case, the above conditions are equivalent to that
\eqref{constant Renyi char} holds as stated in Lemma \ref{lemma:Dalpha monotone}, i.e., 
the
characterization in \cite[Lemma 3.2]{HMO2}. A similar condition to the latter is unknown in the von Neumann algebra case.

Here, for readers' convenience we give a sketchy proof of the above equivalence. To see that (i)$\iff$(ii),
it suffices to assume that $\rho(1)=1$. Since $s(\rho)\le s(\sigma)$ means $\psi(1)=0$ in this case,
the equivalence of (i) and (ii) is immediate. To see that (i)$\iff$(iii), we use Haagerup's $L^p$-spaces
$L^p(\M)$ and Connes' cocycle derivative $(D\rho:D\sigma)_t$. Note that the standard form of $\M$
is given as $(\M,L^2(\M),J=\,^*,L^2(\M)_+)$ and $\M_*$ is order-isomorphic to $L^1(\M)$ by a linear
bijection $\omega\in\M\mapsto h_\omega\in L^1(\M)$, so the vector representative of $\sigma$ is
$h_\sigma^{1/2}\in L^2(\M)_+$. Assume (iii); then $h_\rho=\kappa s(\rho)h_\sigma$ and hence by
\cite[(10.9)]{Hiai_Lectures2021},
\[
\Delta_{\rho,\sigma}^{\alpha/2}h_\sigma^{1/2}=h_\rho^{\alpha/2}h_\sigma^{(1-\alpha)/2}
=\kappa^{\alpha/2}s(\rho)h_\sigma^{1/2},\qquad\alpha\in[0,1].
\]
Therefore, $Q_\alpha(\rho\|\sigma)=\kappa^\alpha\sigma(s(\rho))$ for all $\alpha\in[0,1]$ so that
(i) follows. Conversely assume (i). By (ii) of Lemma \ref{L-C.3} one has
$E_{\rho,\sigma}((0,+\infty)\setminus\{\kappa\})h_\sigma^{1/2}=0$ so that by
\cite[(10.19)]{Hiai_Lectures2021},
\[
(D\rho:D\sigma)_th_\sigma^{1/2}=\Delta_{\rho,\sigma}^{it}h_\sigma^{1/2}
=\kappa^{it}s(\rho)h_\sigma^{1/2},\qquad t\in\bR.
\]
Since $h_\sigma^{1/2}$ is separating for $s(\sigma)\M s(\sigma)$, this means that
$(D\rho:D\sigma)_t=\kappa^{it}s(\rho)$, $t\in\bR$, from which (iii) can be verified.
\end{rem}

\end{document}